\documentclass[10pt]{llncs}

\usepackage{microtype}

\usepackage{wrapfig}
\usepackage{microtype}
\usepackage{graphicx}
\usepackage{graphics}
\usepackage{epsfig}
\usepackage{epstopdf}
\usepackage{amsmath}
\usepackage{latexsym}
\usepackage{wrapfig}
\usepackage{multirow}
\usepackage{amsfonts}
\usepackage{cite}
\usepackage{amssymb}
\usepackage{caption}
\usepackage{algorithm}
\usepackage{algorithmic}
\usepackage{enumerate}
\usepackage{diagbox}
\usepackage{subfig}
\usepackage{wrapfig}
\usepackage{color}
\usepackage{hyperref}

\usepackage[a4paper,margin=2.5cm,footskip=0.25in]{geometry}

\makeatletter
\newcommand{\rmnum}[1]{\romannumeral #1}
\newcommand{\Rmnum}[1]{\expandafter\@slowromancap\romannumeral #1@}
\makeatother

\spnewtheorem{claim}{Claim}{\bfseries}{\rmfamily}

\spnewtheorem{remark}{Remark}{\bfseries}{\rmfamily}

\spnewtheorem{property}{Property}{\bfseries}{\rmfamily}

\begin{document}

\title{A Sub-linear Time Framework for Geometric Optimization with Outliers in High Dimensions}
\author{Hu Ding}
\institute{
 School of Computer Science and Engineering, University of Science and Technology of China \\
 He Fei, China\\
  \email{huding@ustc.edu.cn}\\
}
%
\maketitle

\thispagestyle{empty}


\begin{abstract}
Many real-world problems can be formulated as geometric optimization problems in high dimensions, especially in the fields of machine learning and data mining. Moreover, we often need to take into account of outliers when optimizing the objective functions. However, the presence of outliers could make the problems to be much more challenging than their vanilla versions. In this paper, we study the fundamental minimum enclosing ball (MEB) with outliers problem first; partly inspired by the core-set method from B\u{a}doiu and Clarkson, we propose a sub-linear time bi-criteria approximation algorithm based on two novel techniques, the Uniform-Adaptive Sampling method and Sandwich Lemma. To the best of our knowledge, our result is the first sub-linear time algorithm, which has the sample size ({\em i.e.,} the number of sampled points) independent of both the number of input points $n$ and dimensionality $d$, for MEB with outliers in high dimensions. Furthermore, we observe that these two techniques can be generalized to deal with a broader range of geometric optimization problems with outliers in high dimensions, including flat fitting, $k$-center clustering, and SVM with outliers, and therefore achieve the sub-linear time algorithms for these problems respectively. 
\end{abstract}

%
%
%
  
%

\newpage

\pagestyle{plain}
\pagenumbering{arabic}
\setcounter{page}{1}

\section{Introduction}
\label{sec-intro}
Geometric optimization is a fundamental topic that has been extensively studied in the community of computational geometry~\cite{DBLP:series/lncs/AgarwalEF19}. The {\em minimum enclosing ball (MEB)} problem is one of the most popular geometric optimization problems who has attracted a lot of attentions in past years, where the goal is to compute the smallest ball covering a given set of points in the Euclidean space~\cite{badoiu2003smaller,DBLP:journals/jea/KumarMY03,DBLP:conf/esa/FischerGK03}. Though its formulation is very simple, MEB has a number of applications in real world, such as classification~\cite{DBLP:journals/jmlr/TsangKC05,C10,DBLP:journals/jacm/ClarksonHW12}, preserving privacy~\cite{DBLP:conf/pods/NissimSV16,DBLP:conf/ipsn/FeldmanXZR17}, and quantum cryptography~\cite{gyongyosi2014geometrical}. A more general geometric optimization problem is called {\em flat fitting} that is to compute the smallest slab (centered at a low-dimensional flat) to cover the input data~\cite{DBLP:journals/dcg/Har-PeledV04,panigrahy2004minimum,DBLP:journals/algorithmica/YuAPV08}. Another closely related important topic is the {\em $k$-center clustering} problem, where the goal is to find $k>1$ balls to cover the given input data and minimize the maximum radius of the balls~\cite{gonzalez1985clustering}; the problem has been widely applied to many areas, such as facility location~\cite{charikar2001algorithms} and data analysis~\cite{tan2006introduction}. Moreover, some geometric optimization problems are trying to maximize their size functions. As an example, the well known classification technique {\em support vector machine (SVM)}~\cite{journals/tist/ChangL11} is to maximum the margin separating two differently labeled point sets in the space. 
 
Real-world datasets are often noisy and contain outliers. Moreover, outliers could seriously affect the final optimization results. For example, it is easy to see that even one outlier could make the MEB arbitrarily large. In particular, as the rapid development of machine learning, the field of {\em  adversarial machine learning} concerning about the potential vulnerabilities of the algorithms has attracted a great amount of attentions~\cite{huang2011adversarial,kurakin2016adversarial,biggio2018wild,DBLP:journals/cacm/GoodfellowMP18}. A small set of outliers could be added by some adversarial attacker to make the decision boundary severely deviate and cause unexpected mis-classification~\cite{DBLP:conf/icml/BiggioNL12,DBLP:conf/sp/JagielskiOBLNL18}. Furthermore, the presence of outliers often results in a quite challenging combinatorial optimization problem; as an example, if $m$ of the input $n$ data items are outliers ($m<n$), we have to consider an exponentially large number ${n\choose m}$ of  different possible cases when optimizing the objective function. Therefore, the design of efficient and robust optimization algorithms is urgently needed to meet these challenges.

 \vspace{-0.1in}
 
\subsection{Our Contributions} 

In big data era, the data size could be so large that we cannot even afford to read the whole dataset once. 
In this paper, we consider to develop sub-linear time algorithms for several geometric optimization problems involving outliers. We study the aforementioned MEB with outliers problem first. Informally speaking, given a set of $n$ points in $\mathbb{R}^d$ and a small parameter $\gamma\in(0,1)$, the problem is to find the smallest ball covering at least $(1-\gamma)n$ points from the input. 
We are aware of several existing sub-linear time bi-criteria approximation algorithms based on uniform sampling for MEB and $k$-center clustering with outliers~\cite{alon2003testing,DBLP:conf/focs/HuangJLW18, DBLP:journals/corr/abs-1901-08219}, where the ``bi-criteria'' means that the ball (or the union of the $k$ balls) is allowed to exclude a little more points than the pre-specified number of outliers. Their ideas are based on the theory of VC dimension~\cite{vapnik2015uniform}. But the sample size usually depends on the dimensionality $d$, which is roughly $O\big(\frac{1}{\delta^2\gamma}kd\cdot polylog(\frac{kd}{\delta\gamma})\big)$, if allowing to discard $(1+\delta)\gamma n$ outliers with $\delta\in (0,1)$ ($k=1$ in the complexity for the MEB with outliers problem). 
A detailed overview on previous works is shown in Section~\ref{sec-related}.

Since many optimization problems in practice need to consider high-dimensional datasets, especially in the fields of machine learning and data mining, the above sample size from~\cite{alon2003testing,DBLP:conf/focs/HuangJLW18, DBLP:journals/corr/abs-1901-08219} could be very large. 
Partly inspired by the core-set method from B\u{a}doiu and Clarkson~\cite{badoiu2003smaller} for computing MEB in high dimensions, \textbf{we are wondering that whether it is possible to remove the dependency on $d$ in the sample size for MEB with outliers and other related high dimensional geometric optimization problems}. Given a parameter $\epsilon\in (0,1)$, the method of~\cite{badoiu2003smaller} is a simple greedy algorithm that selects $\frac{2}{\epsilon}$ points (as the core-set) for constructing a $(1+\epsilon)$-approximate MEB, where the resulting radius is at most $1+\epsilon$ times the optimal one. A highlight of their method is that the core-set size $\frac{2}{\epsilon}$ is independent of $d$. However, there are several substantial challenges when applying their method to design sub-linear time algorithm for MEB with outliers. First, we need to implement the ``greedy selection'' step by a random sampling manner, but it is challenging to guarantee the resulting quality especially when the data is mixed with outliers. Second, the random sampling approach often yields a set of candidates for the ball center ({\em e.g.,} we may need to repeatedly run the algorithm multiple times for boosting the success probability, and each time generates a candidate solution), and thus it is necessary to design an efficient strategy to determine which candidate is the best one in sub-linear time. 
 
 To tackle these challenges, we propose two key techniques, the novel ``\textbf{Uniform-Adaptive Sampling}'' method and ``\textbf{Sandwich Lemma}''. Roughly speaking, the Uniform-Adaptive Sampling method can help us to bound the error induced in each ``randomized greedy selection'' step; the Sandwich Lemma enables us to estimate the objective value of each candidate and select the best one in sub-linear time.  
 To the best of our knowledge, our result is the first sub-linear time approximation algorithm for the MEB with outliers problem with sample size being independent of the number of points $n$ and the dimensionality $d$,  which significantly improves the time complexities of existing algorithms.  

Moreover, we observe that our proposed techniques can be used to solve a broader range of geometric optimization problems. We define a general optimization problem called \textbf{ minimum enclosing ``x'' (MEX) with Outliers}, where the ``x'' could be a specified kind of shape ({\em e.g.,} the shape is a ball for MEB with outliers). We prove that it is able to generalize the Uniform-Adaptive Sampling method and  Sandwich Lemma to adapt the shape ``x'', as long as it satisfies several properties. 
 In particular we focus on the MEX with outlier problems including flat fitting, $k$-center clustering, and SVM with outliers; a common characteristic of these problems is that each of them has an iterative algorithm based on greedy selection for its vanilla version (without outliers) that is similar to the MEB algorithm of~\cite{badoiu2003smaller}. Though these problems have been widely studied before, the research in terms of their sub-linear time algorithms is till quite limited. 

\vspace{-0.1in}
\subsection{Related Work}
\label{sec-related}
\vspace{-0.05in}
\textbf{Sub-linear time algorithms.} It has a long history of the research on sub-linear time algorithms design in theory~\cite{rubinfeld2006sublinear,czumaj2006sublinear}. For example, a number of sub-linear time clustering algorithms have been studied in~\cite{DBLP:conf/stoc/Indyk99,meyerson2004k,mishra2001sublinear,czumaj2004sublinear}. Another important application of sub-linear time algorithms is property testing on graphs or probability distributions~\cite{DBLP:journals/jacm/GoldreichGR98}. 

As mentioned before, the uniform sampling idea can be used to design sub-linear time algorithms for the problems of MEB and $k$-center clustering with outliers~\cite{alon2003testing,DBLP:conf/focs/HuangJLW18, DBLP:journals/corr/abs-1901-08219}, but the sample size depends on the dimensionality $d$ that could be very large in practice. Note that Alon {\em et al.}~\cite{alon2003testing} presented another sub-linear time algorithm, which has the sample size independent of $d$,  to test whether an input point set can be covered by a ball with a given radius; however, it is difficult to apply their method to solve the MEB with outliers problem as the algorithm relies on some nice properties of minimum enclosing ball, but these properties are not easy to be utilized when inliers and outliers are mixed. In~\cite{DBLP:journals/corr/abs-1904-03796-2}, we proposed a notion of stability for MEB and developed the sub-linear time MEB algorithms for stable instance. Clarkson {\em et al.}~\cite{DBLP:journals/jacm/ClarksonHW12} developed an elegant perceptron framework for solving several optimization problems arising in machine learning, such as MEB.  For a set of $n$ points in $\mathbb{R}^d$, 
their framework can solve the MEB problem in $\tilde{O}(\frac{n}{\epsilon^2}+\frac{d}{\epsilon})$~\footnote{The asymptotic notation $\tilde{O}(f)=O\big(f\cdot polylog(\frac{nd}{\epsilon})\big)$.} time. Based on a stochastic primal-dual approach, Hazan {\em et al.}~\cite{DBLP:conf/nips/HazanKS11} provided an algorithm for solving the SVM problem in sub-linear time.

%

\textbf{MEB and $k$-center clustering with outliers.}  {\em Core-set} is a popular technique to reduce the time complexities for many optimization problems~\cite{agarwal2005geometric,DBLP:journals/corr/Phillips16}. The core-set idea  has also been used to compute approximate MEB  in high dimensional space~\cite{C10,BHI,DBLP:journals/jea/KumarMY03,panigrahy2004minimum,DBLP:conf/isaac/KerberS13}. B\u{a}doiu and Clarkson \cite{badoiu2003smaller} showed that it is possible to find a core-set of size $\lceil2/\epsilon\rceil$ that yields a $(1+\epsilon)$-approximate MEB. 
There are also several exact and approximation algorithms for MEB that do not rely on core-sets~\cite{DBLP:conf/esa/FischerGK03,DBLP:conf/soda/SahaVZ11,DBLP:conf/icalp/ZhuLY16}. Streaming algorithms for computing MEB were also studied before~\cite{DBLP:journals/algorithmica/AgarwalS15,DBLP:journals/comgeo/ChanP14}.

B\u{a}doiu {\em et al.}~\cite{BHI} extended their core-set idea to the problems of MEB and $k$-center clustering with outliers, and achieved linear time bi-criteria approximation algorithms (if $k$ is assumed to be a constant). 
Several algorithms for the low dimensional MEB with outliers problem have also been developed~\cite{aggarwal1991finding,efrat1994computing,har2005fast,matouvsek1995enclosing}.
A $3$-approximation algorithm for $k$-center clustering with outliers in arbitrary metrics was proposed by Charikar {\em et al.}~\cite{charikar2001algorithms}; Chakrabarty {\em et al.}~\cite{DBLP:conf/icalp/ChakrabartyGK16} proposed a $2$-approximation algorithm for $k$-center clustering with outliers. These algorithms often have high time complexities ({\em e.g.,} $\Omega(n^2d)$). Recently, Ding {\em et al.}~\cite{DBLP:journals/corr/abs-1901-08219} provided a linear time greedy algorithm for $k$-center clustering with outliers based on the idea of the Gonzalez's algorithm~\cite{gonzalez1985clustering}. 
 Furthermore, there exist a number of works on streaming and distributed algorithms, such as~\cite{charikar2003better,mccutchen2008streaming,zarrabistreaming,malkomes2015fast,guha2017distributed,DBLP:journals/corr/abs-1802-09205,li2018distributed}.

\textbf{Flat fitting with outliers.} Given an integer $j\geq 0$ and a set of points in $\mathbb{R}^d$, the flat fitting problem is to find a $j$-dimensional flat having the smallest maximum distance to the input points~\cite{DBLP:conf/focs/Har-PeledV01}; obviously, the MEB problem is a special case with $j=0$.  In high dimensions, Har-Peled and Varadarajan~\cite{DBLP:journals/dcg/Har-PeledV04} provided a linear time algorithm if $j$ is assumed to be fixed; their running time was further reduced by Panigrahy~\cite{panigrahy2004minimum} based on a core-set approach. There also exist several methods considering flat fitting with outliers but only for low-dimensional case~\cite{har2004shape,agarwal2008robust}.

\textbf{SVM with outliers.}  Given two point sets $P_1$ and $P_2$ in $\mathbb{R}^d$, the problem of {\em Support Vector Machine (SVM)} is to find the largest margin to separate $P_1$ and $P_2$ (if they are separable)~\cite{journals/tist/ChangL11}. 
SVM can be formulated as a quadratic programming problem, and a number of efficient techniques have been developed in the past, such as the soft margin SVM~\cite{mach:Cortes+Vapnik:1995,platt99}, $\nu$-SVM~\cite{bb57389,conf/nips/CrispB99}, and Core-SVM~\cite{tkc-cvmfstv-05}. There also exist a number of works on designing robust algorithms of SVM with outliers~\cite{conf/aaai/XuCS06,icml2014c2_suzumura14,ding2015random}.

\section{Definitions and Preliminaries}
\label{sec-pre}

In this paper, we let $|A|$ denote the number of points of a given point set $A$ in $\mathbb{R}^d$, and $||x-y||$ denote the Euclidean distance between two points $x$ and $y$ in $\mathbb{R}^d$. We use $\mathbb{B}(c, r)$ to denote the ball centered at a point $c$ with radius $r>0$. Below, we give several definitions used throughout this paper. 

\begin{definition}[Minimum Enclosing Ball (MEB)]
\label{def-meb}
Given a set $P$ of $n$ points in $\mathbb{R}^d$, the MEB problem is to find a ball with minimum radius to cover all the points in $P$. The resulting ball and its radius are denoted by $MEB(P)$ and $Rad(P)$, respectively.
\end{definition}

A ball $\mathbb{B}(c, r)$ is called a $\lambda$-approximation of $MEB(P)$ for some $\lambda\geq 1$, if the ball covers all points in $P$ and has radius $r\leq \lambda Rad(P)$.

\begin{definition}[MEB with Outliers]
\label{def-outlier}
Given a set $P$ of $n$ points in $\mathbb{R}^d$ and a small parameter $\gamma\in (0,1)$, the MEB with outliers problem is to find the smallest ball that covers $(1-\gamma)n$ points. Namely, the task is to find a subset of $P$ with size $(1-\gamma)n$ such that the resulting MEB is the smallest among all possible choices of the subset. The obtained ball is denoted by $MEB(P, \gamma)$.
\end{definition}

For convenience, we  use $P_{\textnormal{opt}}$ to denote the optimal subset of $P$ with respect to $MEB(P, \gamma)$. That is,
$P_{\textnormal{opt}}=\arg_{Q}\min\Big\{Rad(Q)\mid Q\subset P, \left|Q\right|= (1-\gamma)n\Big\}$. 
From Definition~\ref{def-outlier}, we can see that the main issue is to determine the subset of $P$. Actually, solving such combinatorial problems involving outliers are often 
challenging. 
In Section~\ref{sec-lower}, we will present an example to show that it is impossible to achieve an approximation factor less than $2$ for MEB with outliers, if the time complexity is required to be independent of $n$. Therefore, we consider finding the bi-criteria approximation. Actually, it is also a common way for solving other optimization problems with outliers. 
For example,  Mount {\em et al.}~\cite{DBLP:journals/algorithmica/MountNPSW14} and Meyerson {\em et al.}~\cite{DBLP:journals/ml/MeyersonOP04} studied the bi-criteria approximation algorithms respectively for the problems of linear regression and $k$-median clustering with outliers before.  
%
%
%
%
%
%
%
%
%

\begin{definition}[Bi-criteria Approximation]
\label{def-app}
Given an instance $(P, \gamma)$ for MEB with outliers and two small parameters $0<\epsilon, \delta<1$, a $(1+\epsilon, 1+\delta)$-approximation of $(P, \gamma)$ is a ball that covers at least $\big(1-(1+\delta)\gamma\big)n$ points and has radius at most $(1+\epsilon)Rad(P_{opt})$.
\end{definition}
When both $\epsilon$ and $\delta$ are small, the bi-criteria approximation is very close to the optimal solution with only slight changes on the number of covered points and the radius.

We also extend Definition~\ref{def-outlier} to the problem called \textbf{ minimum enclosing ``x'' (MEX) with Outliers}, where the ``x'' could be any specified shape. To keep the structure of our paper more compact, we state the formal definition of MEX with outliers and the corresponding results in Section~\ref{sec-ext}.

\subsection{A More Careful Analysis for Core-set Construction in~\cite{badoiu2003smaller}}
\label{sec-newanalysis}

Before presenting our main results, we first revisit the core-set construction algorithm for MEB of B\u{a}doiu and Clarkson~\cite{badoiu2003smaller}, since their method will be  used in our algorithms for MEB with outliers. 
 
Let $0<\epsilon<1$. The algorithm in~\cite{badoiu2003smaller} yields an MEB core-set of size $2/\epsilon$ (for convenience, we always assume that $2/\epsilon$ is an integer). However, there is a small issue in their paper. The analysis assumes that the exact MEB of the core-set is computed in each iteration, but in fact 
one may only compute an approximate MEB. Thus, an immediate question is whether the quality is still guaranteed with such a change. Kumar {\em et al.}~\cite{DBLP:journals/jea/KumarMY03} fixed this issue, and showed that computing a $(1+O(\epsilon^2))$-approximate MEB for the core-set in each iteration still guarantees a core-set with  size  $O(1/\epsilon)$, where the hidden constant is larger than $80$. Clarkson~\cite{C10} systematically studied the {\em Frank-Wolfe} algorithm~\cite{frank1956algorithm}, and showed that the greedy core-set construction algorithm of MEB, as a special case of  the  Frank-Wolfe algorithm, yields a core-set with size slightly larger than $ 4/\epsilon$. Note that there exist several other methods yielding even lower core-set size~\cite{coresets1,DBLP:conf/isaac/KerberS13}, but their construction algorithms are more complicated and thus not applicable to our problems. 
Increasing the core-set size from $2/\epsilon$ to $\alpha/\epsilon$ (for some $\alpha>2$) is neglectable in asymptotic analysis. But in Section~\ref{sec-outlier-general}, we will show that it could cause serious issue if outliers exist. Hence, a core-set of size 
 $2/\epsilon$ is still desirable. \textbf{For this purpose, we provide a new analysis which is also interesting in its own right.}

For the sake of completeness, we first briefly introduce the idea of the core-set construction algorithm in~\cite{badoiu2003smaller}.
Given a point set $P\subset\mathbb{R}^d$, the algorithm is a simple iterative procedure. Initially, it selects an arbitrary point from $P$ and places it into an initially empty set $T$. 
In each of the following $2/\epsilon$ iterations, the algorithm updates the center of $MEB(T)$ and adds to $T$ the farthest point from the current center of $MEB(T)$. 
Finally, the center of $MEB(T)$ induces a $(1+\epsilon)$-approximation for $MEB(P)$. The selected set of $2/\epsilon$ points ({\em i.e.}, $T$) is called the core-set of MEB. To ensure the expected improvement in each iteration, ~\cite{badoiu2003smaller} showed that the following two inequalities hold if the algorithm always selects the farthest point to the current center of $MEB(T)$:
\begin{eqnarray}
r_{i+1}  \geq  (1+\epsilon)Rad(P)-L_i; \hspace{0.2in} r_{i+1} \geq  \sqrt{r^2_i+L^2_i},\label{for-bc2}
 \end{eqnarray}
where $r_i$ and $r_{i+1}$ are the radii of $MEB(T)$ in the $i$-th and $(i+1)$-th iterations, respectively, and $L_i$ is the shifting distance of the center of $MEB(T)$ from the $i$-th to $(i+1)$-th iteration.



\begin{wrapfigure}{r}{0.27\textwidth}
  \vspace{-28pt}
\begin{center}
    \includegraphics[width=0.2\textwidth]{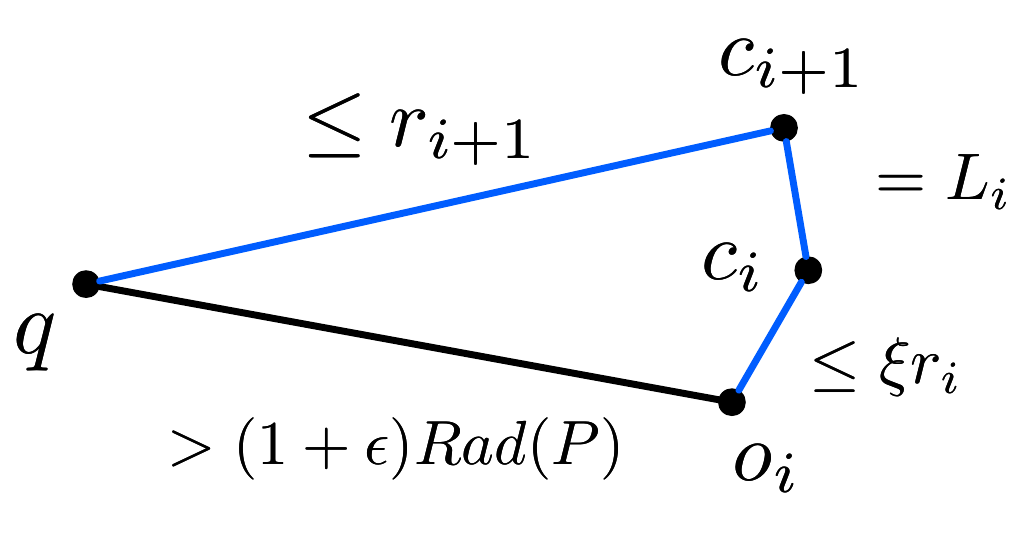}  
    \end{center}
  \vspace{-22pt}
  \caption{An illustration of (\ref{for-bc-1}).}     
   \label{fig-bc}
     \vspace{-25pt}
\end{wrapfigure}

As mentioned earlier, we often compute only an approximate $MEB(T)$ in each iteration. In the $i$-th iteration, we let $c_i$ and $o_i$ denote the centers of the exact and the approximate $MEB(T)$, 
respectively. Suppose that $||c_i-o_i||\leq \xi r_i$, where $\xi\in (0,\frac{\epsilon}{1+\epsilon})$ (we will see why  this bound is needed later). Using another algorithm proposed in~\cite{badoiu2003smaller}, one can obtain the point $o_i$ in $O(\frac{1}{\xi^2}|T|d)$ time. Note that we only compute $o_i$ rather than $c_i$ in each iteration. Hence we can only select the farthest point (say $q$) to $o_i$. If $||q-o_i||\leq (1+\epsilon)Rad(P)$,  we are done and a $(1+\epsilon)$-approximation of MEB is already obtained. Otherwise, we have
\begin{eqnarray}
(1+\epsilon)Rad(P)< ||q-o_i||\leq ||q-c_{i+1}||+||c_{i+1}-c_i||+||c_i-o_i||\leq r_{i+1}+L_i+\xi r_i \label{for-bc-1}
\end{eqnarray}
by  the triangle inequality (see Figure \ref{fig-bc}). In other words, we should replace the first inequality of (\ref{for-bc2}) by $r_{i+1} > (1+\epsilon)Rad(P)-L_i-\xi r_i$. Also, the second inequality of (\ref{for-bc2}) still holds since it depends only on the property of the exact MEB (see Lemma 2.1 in~\cite{badoiu2003smaller}). Thus,  we have 
\begin{eqnarray}
r_{i+1}\geq \max\Big\{\sqrt{r^2_i+L^2_i}, (1+\epsilon)Rad(P)-L_i-\xi r_i\Big\}.\label{for-bc4}
\end{eqnarray}

This leads to the following theorem whose proof can be found 
 in Section~\ref{sec-proof-newbc}.

\begin{theorem}
\label{the-newbc}
In the core-set construction algorithm of~\cite{badoiu2003smaller}, if one computes an approximate MEB for $T$ in each iteration and the resulting center $o_i$ has the distance to $c_i$ less than $\xi r_i= s\frac{\epsilon}{1+\epsilon} r_i$ for some $s\in(0,1)$, the final core-set size is bounded by $z=\frac{2}{(1-s)\epsilon}$. Also, the bound could be arbitrarily close to $2/\epsilon$ when $s$ is small enough. 
\end{theorem}

\begin{remark}
We want to emphasize a simple observation on the above core-set construction procedure, which will be used in our algorithms and analysis later on. The algorithm always selects the farthest point to $o_i$ in each iteration. However, this is actually not necessary. As long as the selected point has distance at least $(1+\epsilon)Rad(P)$, the inequality~(\ref{for-bc-1}) always holds and the following analysis is still true.
 If no such a point exists ({\em i.e.,} $P\setminus \mathbb{B}\big(o_i, (1+\epsilon)Rad(P)\big)=\emptyset$), a $(1+\epsilon)$-approximate MEB ({\em i.e.,} $\mathbb{B}\big(o_i, (1+\epsilon)Rad(P)\big)$) has already been obtained.  
\end{remark}

\section{Two Key Lemmas for Handling Outliers}
\label{sec-twolemma}
In this section, we introduce two important techniques, Lemma~\ref{lem-outlier-general1} and~\ref{lem-outlier-general2}, for solving the problem of MEB with outliers in sub-linear time; the proofs are placed in Section~\ref{sec-proof-outlier-general1} and \ref{sec-proof-outlier-general2}, respectively. The algorithms are presented in Section~\ref{sec-outlier-general}. Moreover, these techniques can be generalized to solve a broader range of optimization problems, and we show the details in Section~\ref{sec-ext}.

To shed some light on our ideas, consider using the core-set construction method in Section~\ref{sec-newanalysis} to compute a bi-criteria $(1+\epsilon, 1+\delta)$-approximation for an instance $(P, \gamma)$ of MEB with outliers. Let $o_i$ be the obtained ball center in the current iteration, and $Q$ be the set of $(1+\delta)\gamma n$ farthest points to $o_i$ from $P$. A key step for updating $o_i$ is finding a point in the set $P_{opt}\cap Q$ (the formal analysis is given in Section~\ref{sec-outlier-general}). Actually, this can be done by performing a random sampling from $Q$. However, it requires to compute the set $Q$ in advance, which takes an $\Omega(nd)$ time complexity. To keep the running time to be sub-linear, we need to find a point from $P_{opt}\cap Q$ by a more sophisticated way. 


Since $P_{opt}$ is mixed with outliers in the set $Q$, simple uniform sampling cannot realize our goal.
To remedy this issue, we propose  a ``two level'' sampling procedure
 which is called ``\textbf{Uniform-Adaptive Sampling}'' (see Lemma~\ref{lem-outlier-general1}). Roughly speaking, we take a random sample $A$ of size $n'$ first ({\em i.e.,} the uniform sampling step), and then randomly select a point from $Q'$, the set of the farthest $\frac{3}{2}(1+\delta)\gamma n'$  points from $A$ to $o_i$ ({\em i.e.,} the adaptive sampling step). According to Lemma~\ref{lem-outlier-general1}, with probability at least $(1-\eta_1) \frac{\delta}{3(1+\delta)}$, the selected point belongs to $P_{opt}\cap Q$; more importantly, the sample size $n'$ is independent of $n$ and $d$.  The key to prove Lemma~\ref{lem-outlier-general1} is to show that the size of the intersection $Q'\cap\big(P_{opt}\cap Q\big)$ is large enough. By setting an appropriate value for $n'$, we can prove a lower bound of $|Q'\cap\big(P_{opt}\cap Q\big)|$. 

\begin{lemma}[Uniform-Adaptive Sampling]
\label{lem-outlier-general1}
Let $\eta_1\in(0,1)$.  If we sample $n'=O(\frac{1}{\delta\gamma}\log\frac{1}{\eta_1})$ points independently and uniformly at random from $P$ and let $Q'$ be the set of farthest $\frac{3}{2}(1+\delta)\gamma n'$ points to $o_i$ from the sample, then, with probability at least $1-\eta_1$, the following holds 
\begin{eqnarray}
\frac{\Big|Q'\cap\big(P_{opt}\cap Q\big)\Big|}{|Q'|}\geq \frac{\delta}{3(1+\delta)}.
\end{eqnarray}
\end{lemma}

\begin{wrapfigure}{r}{0.3\textwidth}
  \vspace{-22pt}
\begin{center}
    \includegraphics[width=0.24\textwidth]{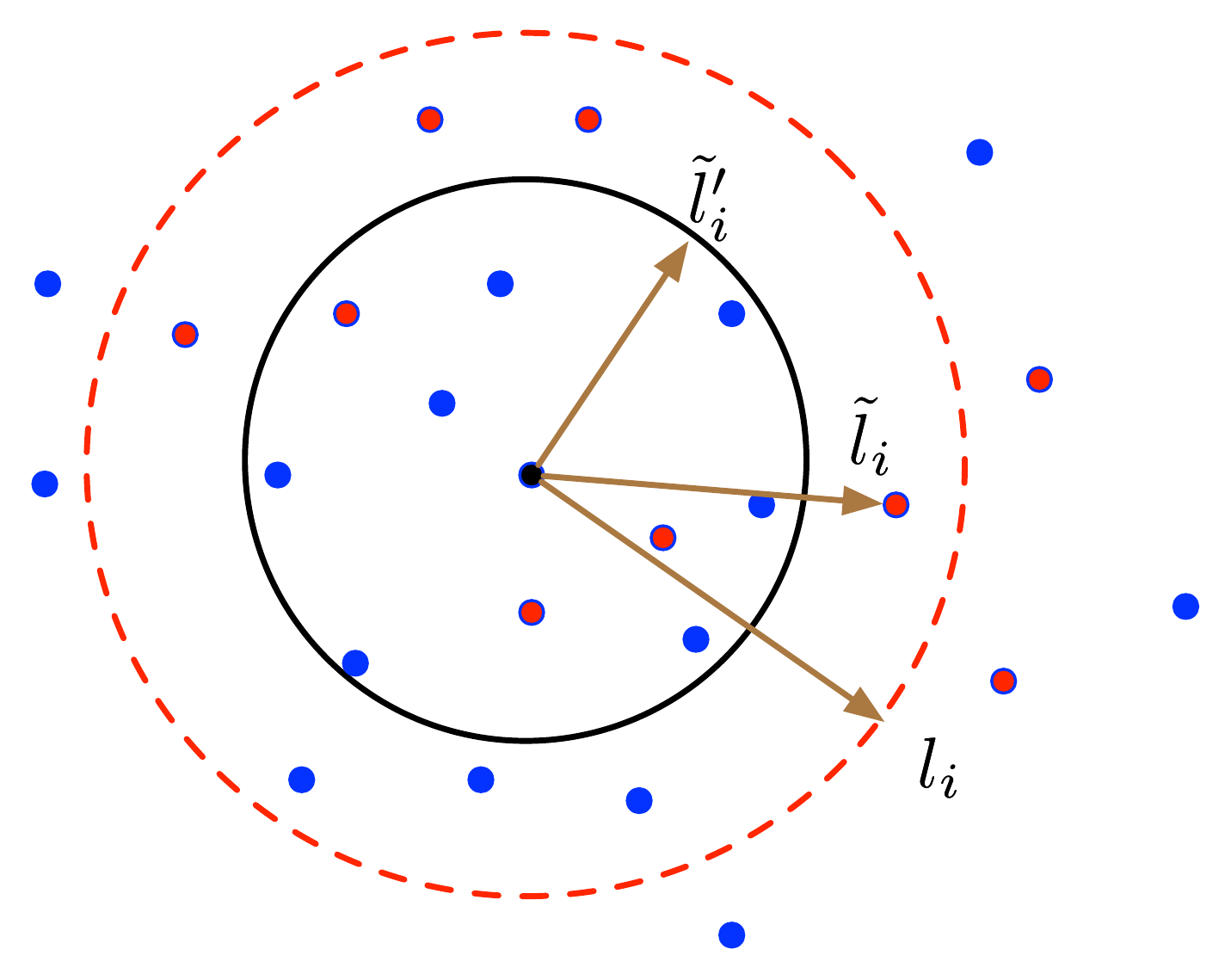}  
    \end{center}
  \vspace{-20pt}
  \caption{The red points are the sampled $n''$ points in Lemma~\ref{lem-outlier-general2}, and the $\big((1+\delta)^2\gamma n''+1\big)$-th farthest point is in the ring bounded by the spheres $\mathbb{B}(o_i, \tilde{l}'_i)$ and $\mathbb{B}(o_i, l_i)$.}     
   \label{fig-sub-outlier}
     \vspace{-22pt}
\end{wrapfigure}

The Uniform-Adaptive Sampling procedure will result in a ``side-effect''. To boost the overall success probability, we have to repeatedly run the algorithm multiple times and each time the algorithm will generate a candidate solution ({\em i.e.,} the ball center). Consequently we have to select the best one as our final solution. With a slight abuse of notation, we still use $o_i$ to denote a candidate ball center; to achieve a $(1+\epsilon, 1+\delta)$-approximation, its radius should be the $\big((1+\delta)\gamma n+1\big)$-th largest distance from $P$ to $o_i$, which is denoted as $l_i$. A straightforward way is to compute the value ``$l_i$'' in linear time for each candidate and return the one having the smallest $l_i$. In this section, we propose the ``\textbf{Sandwich Lemma}'' to estimate $l_i$ in sub-linear time (see Lemma~\ref{lem-outlier-general2}). Let $B$ be the set of $n''$ sampled points from $P$ in Lemma~\ref{lem-outlier-general2}, and $\tilde{l}_i$ be the $\big((1+\delta)^2\gamma n''+1\big)$-th largest distance from  $B$ to $o_i$. The key idea is to prove that the ball $\mathbb{B}(o_i, \tilde{l}_i)$ is ``sandwiched'' by two balls $\mathbb{B}(o_i, \tilde{l}'_i)$ and $\mathbb{B}(o_i, l_i)$, where $\tilde{l}'_i$ is a carefully designed value satisfying (\rmnum{1}) $\tilde{l}'_i\leq \tilde{l}_i\leq l_i$ and (\rmnum{2}) $\Big|P\setminus \mathbb{B}(o_i, \tilde{l}'_i)\Big|\leq (1+O(\delta))\gamma n$. See Figure~\ref{fig-sub-outlier} for an illustration.  These two conditions of $\tilde{l}'_i$ can imply the inequalities (\ref{for-outlier-general2-1}) and (\ref{for-outlier-general2-2}) of Lemma~\ref{lem-outlier-general2}. Further, the inequalities (\ref{for-outlier-general2-1}) and (\ref{for-outlier-general2-2})  jointly imply that $\tilde{l}_i$ is a qualified estimation of $l_i$: if $\mathbb{B}(o_i, l_i)$ is a $(1+\epsilon, 1+\delta)$-approximation, the ball  $\mathbb{B}(o_i, \tilde{l}_i)$  should be a $(1+\epsilon, 1+O(\delta))$-approximation. Similar to Lemma~\ref{lem-outlier-general1}, the sample size $n''$ is also independent of $n$ and $d$.

\begin{lemma} [Sandwich Lemma]
\label{lem-outlier-general2}
Let $\eta_2\in(0,1)$ and assume $\delta<1/3$. If we sample $n''=O\big(\frac{1}{\delta^2\gamma}\log\frac{1}{\eta_2}\big)$ points independently and uniformly at random from $P$ and let $\tilde{l}_i$ be the $\big((1+\delta)^2\gamma n''+1\big)$-th largest distance from the sample to $o_i$, then, with probability $1-\eta_2$, the following holds
\begin{eqnarray}
\tilde{l}_i&\leq& l_i\label{for-outlier-general2-1};\\
\Big|P\setminus \mathbb{B}(o_i, \tilde{l}_i)\Big|&\leq& (1+O(\delta))\gamma n.\label{for-outlier-general2-2}
\end{eqnarray}
\end{lemma}

\subsection{Proof of Lemma~\ref{lem-outlier-general1}}
\label{sec-proof-outlier-general1}

Let $A$ denote  the set of sampled $n'$ points from $P$. First, we know $|Q|=(1+\delta)\gamma n$ and $|P_{opt}\cap Q|\geq \delta\gamma n$ (since there are at most $\gamma n$ outliers in $Q$). For ease of presentation, let $\lambda=\frac{|P_{opt}\cap Q|}{n}\geq \delta\gamma$. Let $\{x_i\mid 1\leq i\leq n'\}$ be $n'$ independent random variables with $x_i=1$ if the $i$-th sampled point of $A$ belongs to $P_{opt}\cap Q$, and $x_i=0$ otherwise. Thus, $E[x_i]= \lambda$ for each $i$. Let $\sigma$ be a small parameter in $(0,1)$. By using the Chernoff bound, we have  $\textbf{Pr}\Big(\sum^{n'}_{i=1}x_i\notin (1\pm\sigma)\lambda n'\Big)\leq e^{-O(\sigma^2 \lambda n')}$. That is,
\begin{eqnarray}
\textbf{Pr}\Big(|A\cap\big(P_{opt}\cap Q\big)|\in (1\pm\sigma)\lambda n'\Big)\geq 1-e^{-O(\sigma^2 \lambda n')}. \label{for-outlier-e3-1}
\end{eqnarray}
Similarly, we have
\begin{eqnarray}
\textbf{Pr}\Big(|A\cap  Q|\in (1\pm\sigma)(1+\delta)\gamma n'\Big)\geq 1-e^{-O(\sigma^2 (1+\delta)\gamma n')}. \label{for-outlier-e4-1}
\end{eqnarray}
Note that $n'=O(\frac{1}{\delta\gamma}\log\frac{1}{\eta_1})$. By setting $\sigma<1/2$ in (\ref{for-outlier-e3-1}) and (\ref{for-outlier-e4-1}), we have
\begin{eqnarray}
\Big|A\cap\big(P_{opt}\cap Q\big)\Big|> \frac{1}{2}\delta\gamma n' \text{\hspace{0.2in} and \hspace{0.2in}} \Big|A\cap  Q\Big|< \frac{3}{2}(1+\delta)\gamma n' \label{for-outlier-general1-1}
\end{eqnarray}
with probability $1-\eta_1$. 
Note that $Q$ contains all the farthest $(1+\delta)\gamma n$ points to $o_i$. Denote by $l_i$ the $\big((1+\delta)\gamma n+1\big)$-th largest distance from $P$ to $o_i$. Thus 
\begin{eqnarray}
A\cap Q=\{p\in A\mid ||p-o_i||>l_i\}.\label{for-aq}
\end{eqnarray}
Also, since $Q'$ is the set of the farthest $\frac{3}{2}(1+\delta)\gamma n'$ points to $o_i$ from $A$, there exists some $l'_i>0$ such that
\begin{eqnarray}
 Q'=\{p\in A\mid ||p-o_i||>l'_i\}.\label{for-qprime}
\end{eqnarray}
(\ref{for-aq}) and (\ref{for-qprime}) imply that either $(A\cap Q)\subseteq Q'$ or $Q'\subseteq (A\cap Q)$. Since $\big|A\cap  Q\big|< \frac{3}{2}(1+\delta)\gamma n'$ and $|Q'|=\frac{3}{2}(1+\delta)\gamma n'$, we know $\Big(A\cap  Q\Big)\subseteq Q'$. Therefore, 
\begin{eqnarray}
\Big(A\cap\big(P_{opt}\cap Q\big)\Big)=\Big(P_{opt}\cap \big(A\cap  Q\big)\Big)\subseteq Q'.\label{for-v9-1} 
\end{eqnarray}
Obviously, 
\begin{eqnarray}
\Big(A\cap\big(P_{opt}\cap Q\big)\Big)\subseteq  \big(P_{opt}\cap Q\big).\label{for-v9-4} 
\end{eqnarray}
The above (\ref{for-v9-1}) and (\ref{for-v9-4}) together imply  
\begin{eqnarray}
\Big(A\cap\big(P_{opt}\cap Q\big)\Big)\subseteq \Big(Q'\cap\big(P_{opt}\cap Q\big)\Big).\label{for-v9-2} 
\end{eqnarray}
Moreover, since $Q'\subseteq A$, we have 
\begin{eqnarray}
\Big(Q'\cap\big(P_{opt}\cap Q\big)\Big)\subseteq \Big(A\cap\big(P_{opt}\cap Q\big)\Big).\label{for-v9-3} 
\end{eqnarray}
Consequently, (\ref{for-v9-2}) and (\ref{for-v9-3}) together imply $Q'\cap\big(P_{opt}\cap Q\big)=A\cap\big(P_{opt}\cap Q\big)$ and hence 
\begin{eqnarray}
\frac{\Big|Q'\cap\big(P_{opt}\cap Q\big)\Big|}{|Q'|}&=&\frac{\Big|A\cap\big(P_{opt}\cap Q\big)\Big|}{|Q'|}\geq \frac{\delta}{3(1+\delta)},
\end{eqnarray}
where the final inequality comes from the first inequality of (\ref{for-outlier-general1-1}) and the fact $|Q'|=\frac{3}{2}(1+\delta)\gamma n'$.

\subsection{Proof of Lemma~\ref{lem-outlier-general2}}
\label{sec-proof-outlier-general2}

Let $B$ denote  the set of sampled $n''$ points from $P$. For simplicity, let $t=(1+\delta)\gamma n$. Assume $\tilde{l}'_i>0$ is the value such that $\Big|P\setminus \mathbb{B}(o_i, \tilde{l}'_i)\Big|=\frac{(1+\delta)^2}{1-\delta}\gamma n$. Recall that $l_i$ is the $\big(t+1\big)$-th largest distance from $P$ to $o_i$. Since $(1+\delta)\gamma n<\frac{(1+\delta)^2}{1-\delta}\gamma n$, it is easy to know $\tilde{l}'_i\leq l_i$. Below, we aim to prove that the $\big((1+\delta)^2\gamma n''+1\big)$-th farthest point from $B$ is in the ring bounded by the spheres $\mathbb{B}(o_i, \tilde{l}'_i)$ and $\mathbb{B}(o_i, l_i)$ (see Figure~\ref{fig-sub-outlier}).

Again, using the Chernoff bound  (let $\sigma=\delta/2$) and the same idea for proving (\ref{for-outlier-general1-1}), since $|B|=n''=O\big(\frac{1}{\delta^2\gamma}\log\frac{1}{\eta_2}\big)$, we have 
\begin{eqnarray}
\Big|B\setminus \mathbb{B}(o_i, \tilde{l}'_i)\Big|&\geq& (1-\delta/2)\frac{(1+\delta)^2}{1-\delta}\gamma n''>(1-\delta)\frac{(1+\delta)^2}{1-\delta}\gamma n''= (1+\delta)^2 \gamma n'' ;\label{for-outlier-general2-3}\\
\Big|B\cap Q\big|&\leq& (1+\delta/2)\frac{t}{n}n''< (1+\delta)\frac{t}{n}n''=(1+\delta)^2 \gamma n'', \label{for-outlier-general2-4}
\end{eqnarray}
with probability $1-\eta_2$. Suppose that (\ref{for-outlier-general2-3}) and (\ref{for-outlier-general2-4}) both hold. Recall that $\tilde{l}_i$ is the $\big((1+\delta)^2\gamma n''+1\big)$-th largest distance from the sampled points $B$ to $o_i$, so $\Big|B\setminus \mathbb{B}(o_i, \tilde{l}_i)\Big|= (1+\delta)^2 \gamma n'' $, and thus $\tilde{l}_i\geq \tilde{l}'_i$ by (\ref{for-outlier-general2-3}). 

The inequality (\ref{for-outlier-general2-4}) implies that the $\big((1+\delta)^2\gamma n''+1\big)$-th farthest point (say $q_x$) from $B$ to $o_i$ is not in $Q$. Then, we claim that $ \mathbb{B}(o_i, \tilde{l}_i)\cap Q=\emptyset$. Otherwise, let $q_y\in  \mathbb{B}(o_i, \tilde{l}_i)\cap Q$. 
Then we have 
\begin{eqnarray}
||q_y-o_i||\leq\tilde{l}_i=||q_x-o_i||.\label{for-v9-5}
\end{eqnarray} 
Note that $Q$ is the set of farthest $t$ points to $o_i$ of $P$. So $q_x\notin Q$ implies 
\begin{eqnarray}
||q_x-o_i||<\min_{q\in Q}||q-o_i|| \leq ||q_y-o_i||
\end{eqnarray}
which is in contradiction to (\ref{for-v9-5}). Therefore, $ \mathbb{B}(o_i, \tilde{l}_i)\cap Q=\emptyset$. 
Further, since  $\mathbb{B}(o_i, l_i)$ excludes exactly the farthest $t$ points ({\em i.e.}, $Q$), $ \mathbb{B}(o_i, \tilde{l}_i)\cap Q=\emptyset$ implies $\tilde{l}_i\leq l_i$.

Overall, we have $\tilde{l}_i\in[\tilde{l}'_i,l_i]$, {\em i.e.,} the $\big((1+\delta)^2\gamma n''+1\big)$-th farthest point from $B$ locates in the ring bounded by the spheres $\mathbb{B}(o_i, \tilde{l}'_i)$ and $\mathbb{B}(o_i, l_i)$ as shown in Figure~\ref{fig-sub-outlier}. Also, $\tilde{l}_i\geq \tilde{l}'_i$ implies 
\begin{eqnarray}
\Big|P\setminus \mathbb{B}(o_i, \tilde{l}_i)\Big|&\leq& \Big|P\setminus \mathbb{B}(o_i, \tilde{l}'_i)\Big|=\frac{(1+\delta)^2}{1-\delta}\gamma n=(1+O(\delta))\gamma n,
\end{eqnarray}
where the last equality comes from the assumption $\delta<1/3$. So (\ref{for-outlier-general2-1}) and (\ref{for-outlier-general2-2}) are true in Lemma~\ref{lem-outlier-general2}. 

\section{Sub-linear Time Algorithm of MEB with Outliers}
\label{sec-outlier-general}
%
%
%
%
%

Recall the remark following Theorem~\ref{the-newbc}. 
%
As long as the selected point has a distance to the center of $MEB(T)$ larger than $(1+\epsilon)$ times the optimal radius, the expected 
improvement will always be guaranteed. Following this observation, we investigate the following approach. 
%
%
%
%
%
%
%
%
%
%
%
Suppose we run the core-set construction procedure decribed in Theorem~\ref{the-newbc} (we should replace $P$ by $P_{opt}$ in our following analysis). 
In the $i$-th step, we add an arbitrary point from $P_{\textnormal{opt}}\setminus \mathbb{B}(o_i, (1+\epsilon)Rad(P_{opt}))$ to $T$. We know that a $(1+\epsilon)$-approximation is obtained after at most $ \frac{2}{(1-s)\epsilon} $ steps, that is, $P_{opt}\subset \mathbb{B}\big(o_i, (1+\epsilon)Rad(P_{opt})\big)$ for some $i\leq   \frac{2}{(1-s)\epsilon}  $. 


However, we need to solve two key issues in order to implement 
the above approach: \textbf{(\rmnum{1})} how to determine the value of $Rad(P_{opt})$ and \textbf{(\rmnum{2})} how to correctly select a point from $P_{\textnormal{opt}}\setminus\mathbb{B}(o_i, (1+\epsilon)Rad(P_{opt}))$. Actually, we can implicitly avoid the first issue via replacing $(1+\epsilon)Rad(P_{opt})$ by the $t$-th largest distance from the points of $P$ to $o_i$, where we set $t=(1+\delta)\gamma n$ for achieving a $(1+\epsilon, 1+\delta)$-approximation in the following analysis. For the second issue, we randomly select one point from the farthest $t$ points of $P$ to $o_i$, and show that it belongs to $P_{\textnormal{opt}}\setminus\mathbb{B}(o_i, (1+\epsilon)Rad(P_{opt}))$ with a certain probability.



Based on the above idea, we present a sub-linear time  $(1+\epsilon, 1+\delta)$-approximation algorithm in this section. To better  understand the algorithm, we show a linear time algorithm first (Algorithm~\ref{alg-outlier} in Sections~\ref{sec-quality}). Note that B\u{a}doiu {\em et al.}~\cite{BHI} also achieved a $(1+\epsilon, 1+\delta)$-approximation algorithm but with a higher complexity. \textbf{Please see more details in our analysis on the running time at the end of Sections~\ref{sec-quality}}. More importantly, we improve the running time of Algorithm~\ref{alg-outlier} to be sub-linear. 
%
%
%
For this purpose, we need to avoid computing the farthest $t$ points to $o_i$, since this operation will take 
linear time. 
Also, Algorithm~\ref{alg-outlier} generates a set of candidates for the solution and we need to select the best one. This process also costs linear time.  By using the techniques proposed in Section~\ref{sec-twolemma}, we can remedy these issues and develop 
%
a sub-linear time algorithm that has the sample complexity independent of $n$ and $d$, in Section~\ref{sec-oulier-improve}.



  \subsection{A Linear Time Algorithm}
\label{sec-quality}

\renewcommand{\algorithmicrequire}{\textbf{Input:}}
\renewcommand{\algorithmicensure}{\textbf{Output:}}
\begin{algorithm}
   \caption{$(1+\epsilon,1+\delta)$-approximation Algorithm for MEB with Outliers}
   \label{alg-outlier}
\begin{algorithmic}[1]
\REQUIRE A point set $P$ with $n$ points in $\mathbb{R}^{d}$, the fraction of outliers $\gamma\in(0,1)$, and the parameters $0<\epsilon,\delta<1$, $z\in\mathbb{Z}^{+}$.
\STATE Let $t=(1+\delta)\gamma n$.
\STATE Initially, randomly select a point $p\in P$ and let $T=\{p\}$. 
\STATE $i=1$; repeat the following steps until $i>z$:
\begin{enumerate}[(1)]
\item Compute the approximate MEB center $o_i$ of $T$.
\item Let $Q$ be the set of farthest $t$ points from $P$ to $o_i$; denote by $l_i$ the $(t+1)$-th largest distance from $P$ to $o_i$.
\item Randomly select a point $q\in Q$, and add it to $T$. 
\item $i=i+1$.
\end{enumerate}
\STATE Output the ball $\mathbb{B}(o_{\hat{i}}, l_{\hat{i}})$ where $\hat{i}=\arg_{i}\min\{l_i\mid 1\leq i\leq z\}$.
\end{algorithmic}
\end{algorithm}

In this section, we present our linear time $(1+\epsilon,1+\delta)$-approximation algorithm for MEB with outliers (see Algorithm~\ref{alg-outlier}). 
For convenience, denote by $c_i$ and $r_i$ the exact center and radius of $MEB(T)$ respectively in the $i$-th round of Step 3 of Algorithm~\ref{alg-outlier}.
In Step 3(1), we compute the approximate center $o_i$ with a distance to $c_i$ of less than $\xi Rad(T)=s\frac{\epsilon}{1+\epsilon}Rad(T)$, where $s\in(0,1)$ as described in Theorem~\ref{the-newbc} (we will determine the value of $s$ in our following analysis on the running time). The following theorem shows the success probability of Algorithm~\ref{alg-outlier}.

\begin{theorem}
\label{the-outlier}
If the input parameter $z= \frac{2}{(1-s)\epsilon}$ (we assume it is an integer for convenience), then with probability $(1-\gamma)(\frac{\delta}{1+\delta})^{z}$, Algorithm~\ref{alg-outlier} outputs a $(1+\epsilon,1+\delta)$-approximation for the MEB with outliers problem.
\end{theorem}

Before proving Theorem~\ref{the-outlier}, we present the following two lemmas first.

\begin{lemma}
\label{lem-outlier-1}
With probability $(1-\gamma)(\frac{\delta}{1+\delta})^{z}$, after running $z$ rounds in Step 3, the set $T\subset P_{opt}$ in Algorithm~\ref{alg-outlier}.
\end{lemma}
\begin{proof}
Initially, because $|P_{opt}|/|P|=1-\gamma$, the first selected point in Step 2 belongs to $P_{opt}$ with probability $1-\gamma$. In each of the $z$ rounds in Step 3, the selected point belongs to $P_{opt}$ with probability $\frac{\delta}{1+\delta}$, since 
\begin{eqnarray}
\frac{|P_{opt}\cap Q|}{|Q|}&=& 1-\frac{|Q\setminus P_{opt}|}{|Q|}\geq 1-\frac{|P\setminus P_{opt}|}{|Q|}=1-\frac{\gamma n}{(1+\delta)\gamma n}=\frac{\delta}{1+\delta}. \label{for-lem-outlier-1}
\end{eqnarray}
Therefore, $T\subset P_{opt}$ with probability $(1-\gamma)(\frac{\delta}{1+\delta})^{z}$.
\qed\end{proof}

\begin{lemma}
In the $i$-th round of Step 3 for $1\leq i\leq z$, at least one of the following two events happens: (1) $o_i$ is the ball center of a $(1+\epsilon,1+\delta)$-approximation; (2) $r_{i+1}>(1+\epsilon)Rad(P_{opt})-||c_i-c_{i+1}||-\xi r_i$. 
\label{lem-outlier-2}
\end{lemma}
\begin{proof}
If $l_i\le(1+\epsilon)Rad(P_{opt})$, then we are done. That is, the ball $\mathbb{B}(o_i, l_i)$ covers $(1-(1+\delta)\gamma)n$ points with radius $l_i\leq (1+\epsilon)Rad(P_{opt})$ (the first event happens). Otherwise, $l_i> (1+\epsilon)Rad(P_{opt})$ and we consider the second event.
Let $q$ be the point added to $T$ in the $i$-th round. Using the triangle inequality, we have
\begin{eqnarray}
||o_i-q||\leq ||o_i-c_i||+||c_i-c_{i+1}||+|c_{i+1}-q||\leq \xi r_i+||c_i-c_{i+1}||+r_{i+1}.\label{for-lemma5-1}
\end{eqnarray}
Since $l_i> (1+\epsilon)Rad(P_{opt})$ and $q$ lies outside  of $\mathbb{B}(o_i, l_i)$, {\em i.e,} $||o_i-q||\geq l_i> (1+\epsilon)Rad(P_{opt})$, (\ref{for-lemma5-1}) implies that the second event happens and the proof is completed.
\qed\end{proof}

\begin{proof} \textbf{(of Theorem~\ref{the-outlier})}
Suppose that the first event of Lemma~\ref{lem-outlier-2} never happens. As a consequence, we obtain a series of inequalities for each pair of radii $ r_{i+1}$ and $r_{i}$, {\em i.e.,} $r_{i+1}>(1+\epsilon)Rad(P_{opt})-||c_i-c_{i+1}||-\xi r_i$. Assume that $T\subset P_{opt}$ in Lemma~\ref{lem-outlier-1}, {\em i.e.,} each time the algorithm correctly adds a point from $P_{opt}$ to $T$. 
Using the almost identical idea for proving Theorem~\ref{the-newbc} in Section~\ref{sec-newanalysis}, 
we know that a $(1+\epsilon)$-approximate MEB of $P_{opt}$ is obtained after at most $z$ rounds. The success probability directly comes from Lemma~\ref{lem-outlier-1}. Overall, we obtain Theorem~\ref{the-outlier}. 
\qed\end{proof}

Moreover, Theorem~\ref{the-outlier} implies the following corollary.

\begin{corollary}
\label{cor-outlier}
If one repeatedly runs Algorithm~\ref{alg-outlier} $O(\frac{1}{1-\gamma}(1+\frac{1}{\delta})^z)$ times, with constant probability, the algorithm outputs a $(1+\epsilon,1+\delta)$-approximation for the problem of MEB with outliers.
\end{corollary}

\textbf{Running time.} 
In Theorem~\ref{the-outlier}, we set $z=\frac{2}{(1-s)\epsilon}$ and $s\in(0,1)$. To keep $z$ small, according to Theorem~\ref{the-newbc}, we set $s=\frac{\epsilon}{2+\epsilon}$ so that $z=\frac{2}{\epsilon}+1$ (only larger than the lower bound $\frac{2}{\epsilon}$ by $1$). 
For each round of Step 3, we need to compute an approximate center $o_i$ that has a distance to the exact one less than $\xi r_i=s\frac{\epsilon}{1+\epsilon}r_i=O(\epsilon^2)r_i$. Using the algorithm proposed in~\cite{badoiu2003smaller}, this can be done in $O(\frac{1}{\xi^2}|T|d)=O(\frac{1}{\epsilon^5}d)$ time. Also, the set $Q$ can be obtained in linear time by the algorithm in~\cite{blum1973time}. 
In total, the time complexity for obtaining a $(1+\epsilon,1+\delta)$-approximation in Corollary~\ref{cor-outlier} is 
\begin{eqnarray}
O\big(\frac{C}{\epsilon}(n+\frac{1}{\epsilon^5})d\big) \label{for-time-bi},
\end{eqnarray}
where $C=O(\frac{1}{1-\gamma}(1+\frac{1}{\delta})^{\frac{2}{\epsilon}+1})$.  As mentioned before, B\u{a}doiu {\em et al.}~\cite{BHI} also achieved a linear time bi-criteria approximation. However, the hidden constant of their running time is exponential in $\Theta(\frac{1}{\epsilon\mu})$ (where $\mu$ is defined in~\cite{BHI}, and should be $\delta\gamma$ to ensure a $(1+\epsilon, 1+\delta)$-approximation) that is much larger than $\frac{2}{\epsilon}+1$.

%

\subsection{Improvement on Running Time}
\label{sec-oulier-improve}

In this section, we show that the running time of Algorithm~\ref{alg-outlier} can be further improved to be independent of the number of points $n$. First, we observe that it is not necessary to compute the set $Q$ of the farthest $t$ points in Step 3(2) of the algorithm. Actually, as long as the selected point $q$ is part of $ P_{opt}\cap Q$ in Step 3(3), a $(1+\epsilon, 1+\delta)$-approximation is still guaranteed. The Uniform-Adaptive Sampling procedure proposed in Section~\ref{sec-twolemma} can help us to obtain a point $q\in P_{opt}\cap Q$ without computing the set $Q$. Moreover, in Lemma~\ref{lem-outlier-general2}, we show that the radius of each candidate solution can be estimated via random sampling. Overall, we achieve a sub-linear time algorithm (Algorithm~\ref{alg-outlier2}). Following the analysis in Section~\ref{sec-quality}, we set $s=\frac{\epsilon}{2+\epsilon}$ so that $z=\frac{2}{(1-s)\epsilon}=\frac{2}{\epsilon}+1$. We present the results in Theorem~\ref{the-outlier2} and Corollary~\ref{cor-outlier2}. Comparing with Theorem~\ref{the-outlier}, we have an extra $(1-\eta_1)(1-\eta_2)$ in the success probability in Theorem~\ref{the-outlier2}, due to the probabilities from Lemmas~\ref{lem-outlier-general1} and \ref{lem-outlier-general2}.

\renewcommand{\algorithmicrequire}{\textbf{Input:}}
\renewcommand{\algorithmicensure}{\textbf{Output:}}
\begin{algorithm}
\vspace{-0.05in}
   \caption{Sub-linear Time $(1+\epsilon,1+O(\delta))$-approximation Algorithm for MEB with Outliers}
   \label{alg-outlier2}
\begin{algorithmic}[1]
\REQUIRE A point set $P$ with $n$ points in $\mathbb{R}^{d}$, the fraction of outliers $\gamma\in(0,1)$, and the parameters $\epsilon,\eta_1, \eta_2\in (0,1)$, $\delta\in (0,1/3)$, and $z\in\mathbb{Z}^{+}$.
\STATE Let $n'=O(\frac{1}{\delta\gamma}\log\frac{1}{\eta_1})$, $n''=O\big(\frac{1}{\delta^2\gamma}\log\frac{1}{\eta_2}\big)$, $t'=\frac{3}{2}(1+\delta)\gamma n'$, and $t''=(1+\delta)^2\gamma n''$.
\STATE Initially, randomly select a point $p\in P$ and let $T=\{p\}$. 
\STATE $i=1$; repeat the following steps until $j=z$:
\begin{enumerate}[(1)]
\item Compute the approximate MEB center $o_i$ of $T$.

\item Sample $n'$ points uniformly at random from $P$, and let $Q'$ be the set of farthest $t'$ points to $o_i$ from the sample.

\item Randomly select a point $q\in Q'$, and add it to $T$. 

\item Sample $n''$ points uniformly at random from $P$, and let $\tilde{l}_i$ be the $(t''+1)$-th largest distance from the sampled points to $o_i$.
\item $i=i+1$.
\end{enumerate}
\STATE Output the ball $\mathbb{B}(o_{\hat{i}}, \tilde{l}_{\hat{i}})$ where $\hat{i}=\arg_{i}\min\{\tilde{l}_i\mid 1\leq i\leq z\}$.
\end{algorithmic}
\end{algorithm}


\begin{theorem}
\label{the-outlier2}
If the input parameter $z= \frac{2}{\epsilon}+1$, then with probability $(1-\gamma)\big((1-\eta_1)(1-\eta_2)\frac{\delta}{3(1+\delta)}\big)^{z}$, Algorithm~\ref{alg-outlier2} outputs a $(1+\epsilon,1+O(\delta))$-approximation for the problem of MEB with outliers.
\end{theorem}

To boost the success probability in Theorem~\ref{the-outlier2}, we need to repeatedly run Algorithm~\ref{alg-outlier2} and output the best candidate. However, we need to be careful on setting the parameters. The success probability in Theorem~\ref{the-outlier2} consists of two parts, $\mathcal{P}_1=(1-\gamma)\big((1-\eta_1)\frac{\delta}{3(1+\delta)}\big)^{z}$ and $\mathcal{P}_2=(1-\eta_2)^{z}$, where $\mathcal{P}_1$ indicates the probability that $\{o_1, \cdots, o_z\}$ contains a qualified candidate, and  $\mathcal{P}_2$ indicates the success probability of Lemma~\ref{lem-outlier-general2} over all the $z$ rounds. Therefore, if we run Algorithm~\ref{alg-outlier2} $N=O(\frac{1}{\mathcal{P}_1})$ times, with constant probability (by taking the union bound), the set of all the generated candidates contains at least one that yields a $(1+\epsilon,1+O(\delta))$-approximation; moreover, to guarantee that we can correctly estimate the resulting radii of all the candidates via the Sandwich Lemma with constant probability, we need to set $\eta_2=O(\frac{1}{zN})$ (because there are $O(zN)$ candidates).

%
%
%

\begin{corollary}
\label{cor-outlier2}
If one repeatedly runs Algorithm~\ref{alg-outlier2} $N=O\Big(\frac{1}{1-\gamma}\big(\frac{1}{1-\eta_1}(3+\frac{3}{\delta})\big)^{z}\Big)$ times with setting $\eta_2=O(\frac{1}{zN})$, with constant probability, the algorithm outputs a $(1+\epsilon,1+O(\delta))$-approximation for the problem of MEB with outliers.
\end{corollary}

The calculation of running time is similar to (\ref{for-time-bi}) in Section~\ref{sec-quality}. We just replace $n$ by $\max\{n', n''\}=O\big(\frac{1}{\delta^2\gamma}\log\frac{1}{\eta_2}\big)=O\big(\frac{1}{\delta^2\gamma}\log(zN)\big)=\tilde{O}\big(\frac{1}{\delta^2\gamma\epsilon}\big)$~\footnote{The asymptotic notation $\tilde{O}(f)=O\big(f\cdot polylog(\frac{1}{\eta_1\delta(1-\gamma)})\big)$.}, and change the value of $C$ to be $O\Big(\frac{1}{1-\gamma}\big(\frac{1}{1-\eta_1}(3+\frac{3}{\delta})\big)^{\frac{2}{\epsilon}+1}\Big)$. So the total running time is independent of $n$. Also, to covert the result from  $(1+\epsilon,1+O(\delta))$-approximation to  $(1+\epsilon,1+\delta)$-approximation, we just need to reduce the value of $\delta$ in the input of Algorithm~\ref{alg-outlier2} appropriately.

\section{The Extension: MEX with Outliers}
\label{sec-ext}



In this section, we extend Definition~\ref{def-outlier} and define a more general problem called \textbf{minimum enclosing ``x'' (MEX) with Outliers}. We observe that the ideas of Lemma~\ref{lem-outlier-general1} and~\ref{lem-outlier-general2} can be further extended to deal with MEX with outliers problems, as long as the shape ``x'' satisfies several properties. 
\textbf{To describe a shape ``x'', we need to clarify three basic concepts: center, size, and distance function.}

Let $\mathcal{X}$ be the set of specified shapes in $\mathbb{R}^d$. In this paper, we require that each shape $x\in \mathcal{X}$ is uniquely determined by the following two components: 
``$c(x)$'', the \textbf{center} of $x$, and 
``$s(x)\geq 0$'', the \textbf{size} of $x$. 
For any two shapes $x_1, x_2\in\mathcal{X}$, $x_1=x_2$ if and only if $c(x_1)=c(x_2)$ and $s(x_1)=s(x_2)$. Moreover, given a center $o_0$ and a value $l_0\geq 0$, we use $x(o_0, l_0)$ to denote the shape $x$ with $c(x)=o_0$ and $s(x)=l_0$. 
For different shapes, we have different definitions for the center and size. For example, if $x$ is a ball, $c(x)$ and $s(x)$ should be the ball center and the radius respectively; given $o_0\in \mathbb{R}^d$ and  $l_0\geq 0$,  $x(o_0, l_0)$ should be the ball $\mathbb{B}(o_0, l_0)$. As a more complicated example, consider the $k$-center clustering with outliers problem, which is to find $k$ balls to cover the input point set excluding a certain number of outliers and minimize the maximum radius ({\em w.l.o.g.,} we can assume that the $k$ balls have the same radius). For this problem, the shape ``x'' is a union of $k$ balls in $\mathbb{R}^d$; the center $c(x)$ is the set of the $k$ ball centers and the size $s(x)$ is the radius. 


For any point $p\in \mathbb{R}^d$ and any shape $x\in\mathcal{X}$, we also need to define a \textbf{distance function} $f(c(x), p)$ between the center $c(x)$ and $p$.
 For example, if $x$ is a ball, $f(c(x), p)$ is simply equal to $||p-c(x)||$; if $x$ is a union of $k$ balls with the center $c(x)=\{c_1, c_2, \cdots, c_k\}$, the distance should be $\min_{1\leq j\leq k}||p-c_j||$. Note that the distance function is only for ranking the points to $c(x)$, and not necessary to be non-negative ({\em e.g.,} in Section~\ref{sec-svm}, we define a distance function $f(c(x), p)\leq 0$ for SVM). By using this distance function, we can define the set ``$Q$'' and the value ``$l_i$'' when generalizing Lemma~\ref{lem-outlier-general1} and~\ref{lem-outlier-general2} below. To guarantee their correctnesses, we also require $\mathcal{X}$ to satisfy the following three properties. 

\begin{property}
\label{prop-1}
For any two shapes $x_1\neq x_2\in\mathcal{X}$, if $c(x_1)=c(x_2)$, then
\begin{eqnarray}
s(x_1)\leq s(x_2)\Longleftrightarrow x_1 \text{ is covered by } x_2, 
\end{eqnarray}
where ``$x_1$ is covered by $x_2$'' means ``for any point $p\in \mathbb{R}^d$, $p\in x_1\Rightarrow p\in x_2$''.
\end{property}

\begin{property}
\label{prop-2}
Given any shape $x\in\mathcal{X}$ and any point $p_0\in x$, the set 
\begin{eqnarray}
\{p\mid p\in \mathbb{R}^d \text{ and } f(c(x), p)\leq f(c(x), p_0)\}\subseteq x.
\end{eqnarray}
\end{property}

%

\begin{property}
\label{prop-4}
Given any shape center $o_0$ and any point $p_0\in \mathbb{R}^d$, they together determine a value $r_0=\min\{r\mid r\geq 0, p_0\in x(o_0, r)\}$; that is,  $p_0\in x(o_0, r_0)$ and  $p_0 \notin x(o_0, r)$ for any $r<r_0$. (\textbf{Note:} usually the value $r_0$ is just the distance from $p_0$ to the shape center $o_0$; but for some cases, such as the SVM problem in Section~\ref{sec-svm}, the shape size and distance function have different meanings).
 \end{property}

Intuitively, Property~\ref{prop-1} shows that $s(x)$ defines an order of the shapes sharing the same center $c(x)$. Property~\ref{prop-2} shows that the distance function $f$ defines an order of the points to a given shape center $c(x)$. Property~\ref{prop-4} shows that a center $o_0$ and a point $p_0$ can define a shape just ``touching'' $p_0$. We can take $\mathcal{X}=\{\text{all $d$-dimensional balls}\}$ as an example. For any two concentric balls, the smaller one is always covered by the larger one (Property~\ref{prop-1}); if a point $p_0$ is inside a ball $x$, any point $p$ having the distance $||p-c(x)||\leq ||p_0-c(x)||$ should be inside $x$ too (Property~\ref{prop-2}); also, given a ball center $o_0$ and a point $p_0$,  $p_0\in \mathbb{B}(o_0, ||p_0-o_0||)$ and $p_0\notin \mathbb{B}(o_0, r)$ for any $r<||p_0-o_0||$ (Property~\ref{prop-4}). 
 
%
%
%
Now, we introduce the formal definitions of the MEX with outliers problem and its bi-criteria approximation. 

\begin{definition} [MEX with Outliers]
\label{def-outlier-general}
Suppose the shape set $\mathcal{X}$ satisfies Property~\ref{prop-1}, \ref{prop-2}, and \ref{prop-4}. 
Given a set $P$ of $n$ points in $\mathbb{R}^d$ and a small parameter $\gamma\in (0,1)$, the MEX with outliers problem is to find the smallest shape $x\in \mathcal{X}$ that covers $(1-\gamma)n$ points. Namely, the task is to find a subset of $P$ with size $(1-\gamma)n$ such that its minimum enclosing shape of $\mathcal{X}$ is the smallest among all possible choices of the subset. The obtained solution is denoted by $MEX(P, \gamma)$.
\end{definition}


\begin{definition}[Bi-criteria Approximation]
\label{def-app-general}
Given an instance $(P, \gamma)$ for MEX with outliers and two small parameters $0<\epsilon, \delta<1$, a $(1+\epsilon, 1+\delta)$-approximation of $(P, \gamma)$ is a solution $x\in\mathcal{X}$ that covers at least $\big(1-(1+\delta)\gamma\big)n$ points and has the size at most $(1+\epsilon)s(x_{opt})$, where $x_{opt}$ is the optimal solution.
\end{definition}

It is easy to see that Definition~\ref{def-outlier} of MEB with outliers actually is a special case of Definition~\ref{def-outlier-general}. 
Similar to MEB with outliers, we still use $P_{opt}$, where $P_{opt}\subset P$ and $|P_{opt}|=(1-\gamma)n$, to denote the subset covered by the optimal solution of MEX with outliers. 

%
%
Now, we provide the generalized versions of Lemma~\ref{lem-outlier-general1} and \ref{lem-outlier-general2}. Similar to the core-set construction method in Section~\ref{sec-newanalysis}, we assume that there exists an iterative  algorithm $\Gamma$ to compute MEX (without outliers); actually, this is an important prerequisite to design the sub-linear time algorithms under our framework (we will discuss the iterative algorithms for the MEX with outliers problems considered in our paper, including flat fitting, $k$-center clustering, and SVM, in the appendix). 
In the $i$-th iteration of $\Gamma$, it maintains a shape center $o_i$.  Also, let $Q$ be the set of $(1+\delta)\gamma n$ farthest points from $P$ to $o_i$ with respect to the distance function $f$. First, we need to define the value ``$l_i$'' by $Q$ in the following claim. 

\begin{claim}
\label{claim-prop}
There exists a value $l_i\geq 0$ satisfying $P\setminus x(o_i, l_i)=Q$.
\end{claim}
\begin{proof}
The points of $P$ can be ranked based on their distances to $o_i$. Without loss of generality, let $P=\{p_1, p_2, \cdots, p_n\}$ with $f(o_i, p_1)> f(o_i, p_2)>\cdots > f(o_i, p_n)$ (for convenience, we assume that any two distances are not equal; if there is a tie, we can arbitrarily decide their order to $o_i$). Therefore, the set $Q=\{p_j\mid 1\leq j\leq (1+\delta)\gamma n\}$. Moreover, from Property~\ref{prop-4}, we know that each point $p_j\in P$ corresponds to a value $r_j$ that $p_j\in x(o_i, r_j)$ and $p_j\notin x(o_i, r)$ for any $r<r_j$. 
Denote by $x_{j}$ the shape $x(o_i, r_{j})$. We select the  point  $p_{j_0}$ with $j_0=(1+\delta)\gamma n+1$. 
From Property~\ref{prop-2}, we know that $p_j\in x_{j_0}$ for any $j\geq j_0$, {\em i.e.,} \textbf{(a)}~$P\setminus Q\subseteq x_{j_0}$. We also need to prove that $p_j\notin x_{j_0}$ for any $j< j_0$. Assume there exists some $p_{j_1}\in x_{j_0}$ with $j_1< j_0$. Then we have $r_{j_1}< r_{j_0}$ and thus $p_{j_0}\notin x_{j_1}$ (by Property~\ref{prop-4}). By Property~\ref{prop-2}, $p_{j_0}\notin x_{j_1}$ implies $f(o_i, p_{j_0})>f(o_i, p_{j_1})$, which is in contradiction to the fact $f(o_i, p_{j_0})<f(o_i, p_{j_1})$. So we have \textbf{(b)}~$Q\cap x_{j_0}=\emptyset$. 

From the above  \textbf{(a)} and \textbf{(b)}, we know $P\setminus x_{j_0}=Q$. Therefore, we can set the value $l_i=r_{j_0}$ and then $P\setminus x(o_i, l_i)=Q$. 
%
%
%
%
%
%
%
%
%
%
%
%
%
%
%
%
%
%
%
%
%
%
%
%
\qed\end{proof}





\begin{lemma}[Generalized Uniform-Adaptive Sampling]
\label{lem-outlier-general1-general}
Let $\eta_1\in(0,1)$.  If we sample $n'=O(\frac{1}{\delta\gamma}\log\frac{1}{\eta_1})$ points independently and uniformly at random from $P$ and let $Q'$ be the set of farthest $\frac{3}{2}(1+\delta)\gamma n'$ points to $o_i$ from the sample, then, with probability at least $1-\eta_1$, the following holds 
\begin{eqnarray}
\frac{\Big|Q'\cap\big(P_{opt}\cap Q\big)\Big|}{|Q'|}\geq \frac{\delta}{3(1+\delta)}.
\end{eqnarray}
\end{lemma}
\begin{proof}
Let $A$ denote  the set of sampled $n'$ points from $P$. Similar to (\ref{for-outlier-general1-1}), we have
\begin{eqnarray}
\Big|A\cap\big(P_{opt}\cap Q\big)\Big|> \frac{1}{2}\delta\gamma n' \text{\hspace{0.2in} and \hspace{0.2in}} \Big|A\cap  Q\Big|< \frac{3}{2}(1+\delta)\gamma n' \label{for-outlier-general1-1-2}
\end{eqnarray}
with probability $1-\eta_1$. Similar to (\ref{for-aq}), we have 
\begin{eqnarray}
A\cap Q=\{p\in A\mid f(o_i, p)>f(o_i, p_{j_0})\},\label{for-aq-1}
\end{eqnarray}
where $p_{j_0}$ is the point selected in the proof of Claim~\ref{claim-prop}. By using the same manner of Claim~\ref{claim-prop}, we also can select a point $p_{j'_0}\in A$ with 
\begin{eqnarray}
 Q'=\{p\in A\mid  f(o_i, p)>f(o_i, p_{j'_0})\}.\label{for-qprime-1}
\end{eqnarray}
Then, we can prove 
\begin{eqnarray}
\Big(A\cap\big(P_{opt}\cap Q\big)\Big)= \Big(Q'\cap\big(P_{opt}\cap Q\big)\Big).\label{for-v9-2-1} 
\end{eqnarray}
by using the same idea of (\ref{for-v9-2}). Hence, 
\begin{eqnarray}
\frac{\Big|Q'\cap\big(P_{opt}\cap Q\big)\Big|}{|Q'|}&=&\frac{\Big|A\cap\big(P_{opt}\cap Q\big)\Big|}{|Q'|}\geq \frac{\delta}{3(1+\delta)},
\end{eqnarray}
where the final inequality comes from the first inequality of (\ref{for-outlier-general1-1-2}) and the fact $|Q'|=\frac{3}{2}(1+\delta)\gamma n'$.
\qed\end{proof}

\begin{lemma} [Generalized Sandwich Lemma]
\label{lem-outlier-general2-generalize}
Let $\eta_2\in(0,1)$ and assume $\delta<1/3$. $l_i$ is the value from Claim~\ref{claim-prop}. We sample $n''=O\big(\frac{1}{\delta^2\gamma}\log\frac{1}{\eta_2}\big)$ points independently and uniformly at random from $P$ and let $q$ be the $\big((1+\delta)^2\gamma n''+1\big)$-th farthest one from the sampled points to $o_i$. If 
$\tilde{l}_i=\min\{r\mid r\geq 0, q\in x(o_i, r)\}$ (similar to the way defining ``$r_0$'' in Property~\ref{prop-4}), then, with probability $1-\eta_2$, the following holds
\begin{eqnarray}
\tilde{l}_i&\leq& l_i\label{for-outlier-general2-g1};\\
\Big|P\setminus x(o_i, \tilde{l}_i)\Big|&\leq& (1+O(\delta))\gamma n.\label{for-outlier-general2-g2}
\end{eqnarray}
\end{lemma}
\begin{proof}
Let $B$ denote  the set of sampled $n''$ points from $P$. By using the same manner of Claim~\ref{claim-prop}, we know that there exists a value $\tilde{l}'_i>0$ satisfying $\Big|P\setminus x(o_i, \tilde{l}'_i)\Big|=\frac{(1+\delta)^2}{1-\delta}\gamma n$. Similar to the proof of Lemma~\ref{lem-outlier-general2}, we can prove that $\tilde{l}_i\in[\tilde{l}'_i,l_i]$. Due to Property~\ref{prop-1}, we know that $x(o_i, \tilde{l}_i)$ is ``sandwiched'' by the two shapes $x(o_i, \tilde{l}'_i)$ and $x(o_i, l_i)$. Further, since $x(o_i, \tilde{l}'_i)$ is covered by $x(o_i, \tilde{l}_i)$, we have
\begin{eqnarray}
\Big|P\setminus x(o_i, \tilde{l}_i)\Big|&\leq& \Big|P\setminus x(o_i, \tilde{l}'_i)\Big|=\frac{(1+\delta)^2}{1-\delta}\gamma n=(1+O(\delta))\gamma n,
\end{eqnarray}
where the last equality comes from the assumption $\delta<1/3$. So (\ref{for-outlier-general2-g1}) and (\ref{for-outlier-general2-g2}) are true. 
\qed\end{proof}

\textbf{Appendix.} Due to the space limit, we place other parts in our appendix. In Section~\ref{sec-ext-kcenter}, \ref{sec-flat}, and \ref{sec-svm}, we propose the sub-linear time bi-criteria approximation algorithms for three different MEX with outlier problems: $k$-center clustering, flat fitting, and SVM with outliers. All of these problems have important applications in real world, and our results significantly reduce their time complexities comparing with existing approaches. 

\newpage
\bibliographystyle{abbrv}

\bibliography{stability}

\section{The Lower Bound of Sample Size for MEB with Outliers}
\label{sec-lower}
Actually, it is easy to obtain a sub-linear time randomized $2$-approximation algorithm for MEB with outliers (if only returning the center): one can randomly select one point $p\in P$, and it belongs $P_{opt}$ with probability $1-\gamma$; thus, with probability $1-\gamma$, the point $p$ yields a $2$-approximation by using the triangle inequality. Below, we consider an example and show that it is impossible to achieve a single-criterion $\lambda$-approximation for any $\lambda<2$, if the time complexity is required to be independent of $n$. 

\begin{figure}[]
   \centering
  \includegraphics[height=0.7in]{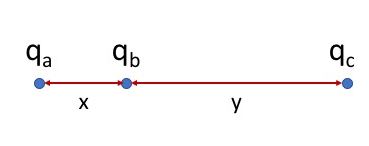}
  \vspace{-0.1in}
      \caption{$||q_a-q_b||=x$ and $||q_b-q_c||=y$.}
  \label{fig-lower}
  \vspace{-0.1in}
\end{figure}

Let $q_a$, $q_b$, and $q_c$ be three colinear points in $\mathbb{R}^d$, where $||q_a-q_b||=x$ and $||q_b-q_c||=y$ (see Figure~\ref{fig-lower}). Let $\gamma\in (0,1)$ and the point set $P=P_a\cup P_b\cup P_c$, where $P_a$ contains a single point located at $q_a$, $P_b$ contains $(1-\gamma)n-1$ points overlapping at $q_b$, and $P_c$ contains $\gamma n$ points overlapping at $q_c$. Consider the instance $(P, \gamma)$ for the problem of MEB with outliers. Suppose $x\ll y$. Consequently, the optimal subset $P_{opt}$ should be $P_a\cup P_b$ and the optimal radius is $x/2$. 

If we take a random sample $S$ of size $m=o(n)$ from $P$, with high probability, $P_a\cap S=\emptyset$ (even if we repeat our sampling a constant number of times, the single point $P_a$ will still be missing with high probability). Therefore, $S$ only contains the points from $P_b$ and $P_c$. If we run an existing algorithm on $S$, the resulting ball center, say $o$, should always be inside the convex hull of $S$, {\em i.e.,} the segment $\overline{q_b q_c}$. Let $\mathbb{B}(o, r)$ be the ball covering $(1-\gamma)n$ points of $P$. We consider two cases: (1) $q_a\in\mathbb{B}(o, r)$ and (2) $q_a\notin\mathbb{B}(o, r)$. For case (1), since $o\in \overline{q_b q_c}$, it is easy to know that the radius $r\geq ||q_a-q_b||=x$ and therefore the approximation ratio is at least $2$. For case (2), since $|P_b|=(1-\gamma)n-1$, $\mathbb{B}(o, r)$ must cover some point from $P_c$ and therefore $r=y/2$; because $x\ll y$, the approximation ratio is also larger than $2$.

\section{Proof of Theorem~\ref{the-newbc}}
\label{sec-proof-newbc}

Similar to the analysis in~\cite{badoiu2003smaller}, we let $\lambda_i=\frac{r_i}{(1+\epsilon)Rad(P)}$. Because $r_i$ is the radius of $MEB(T)$ and $T\subset P$,  we know $r_i\leq Rad(P)$ and then $\lambda_i\leq1/(1+\epsilon)$. By simple calculation, we know that when $L_i=\frac{\big((1+\epsilon)Rad(P)-\xi r_i\big)^2-r^2_i}{2\big((1+\epsilon)Rad(P)-\xi r_i\big)}$ the lower bound of $r_{i+1}$ in (\ref{for-bc4}) achieves the minimum value. Plugging this value of $L_i$ into (\ref{for-bc4}), we have
\begin{eqnarray}
\lambda^2_{i+1}\geq \lambda^2_i+\frac{\big((1-\xi\lambda_i)^2-\lambda^2_i\big)^2}{4(1-\xi\lambda_i)^2}.\label{for-bc5}
\end{eqnarray}
To simplify  inequality (\ref{for-bc5}), we consider the function $g(x)=\frac{(1-x)^2-\lambda^2_i}{1-x}$, where $0<x<\xi$. Its derivative $g'(x)=-1-\frac{\lambda^2_i}{(1-x)^2}$ is always negative, thus we have
\begin{eqnarray}
g(x)\geq g(\xi)=\frac{(1-\xi)^2-\lambda^2_i}{1-\xi}. \label{for-bc2-1}
\end{eqnarray}
Because $\xi<\frac{\epsilon}{1+\epsilon}$ and $\lambda_i\leq 1/(1+\epsilon)$, we know that  the right-hand side of (\ref{for-bc2-1}) is always non-negative. Using (\ref{for-bc2-1}), the inequality (\ref{for-bc5}) can be simplified to 
\begin{eqnarray}
\lambda^2_{i+1}&\geq& \lambda^2_i+\frac{1}{4}\big(g(\xi)\big)^2\nonumber\\
&=&\lambda^2_i+\frac{\big((1-\xi)^2-\lambda^2_i\big)^2}{4(1-\xi)^2}.\label{for-bc2-2}
\end{eqnarray}
(\ref{for-bc2-2}) can be further rewritten as 
\begin{eqnarray}
\Big(\frac{\lambda_{i+1}}{1-\xi}\Big)^2&\geq&\frac{1}{4}\Big(1+(\frac{\lambda_{i}}{1-\xi})^2\Big)^2 \nonumber\\
\Longrightarrow  \frac{\lambda_{i+1}}{1-\xi}&\geq&\frac{1}{2}\Big(1+(\frac{\lambda_{i}}{1-\xi})^2\Big).\label{for-bc2-3}
\end{eqnarray}

Now, we can apply a similar transformation of $\lambda_i$ which was used in~\cite{badoiu2003smaller}. Let $\gamma_i=\frac{1}{1-\frac{\lambda_i}{1-\xi}}$.  We know $\gamma_i>1$ (note $0\leq\lambda_i\leq\frac{1}{1+\epsilon}$ and $\xi<\frac{\epsilon}{1+\epsilon}$). Then, (\ref{for-bc2-3}) implies that 
\begin{eqnarray}
\gamma_{i+1}&\geq&\frac{\gamma_i}{1-\frac{1}{2\gamma_i}}\nonumber\\
&=&\gamma_i\big(1+\frac{1}{2\gamma_i}+(\frac{1}{2\gamma_i})^2+\cdots\big)\nonumber\\
&>&\gamma_i+\frac{1}{2}, \label{for-bc2-4}
\end{eqnarray}
where the equation comes from the fact that $\gamma_i>1$ and thus $\frac{1}{2\gamma_i}\in(0,\frac{1}{2})$. Note that $\lambda_0=0$ and thus $\gamma_0=1$. As a consequence, we have $\gamma_i>1+\frac{i}{2}$. In addition, since $\lambda_i\leq\frac{1}{1+\epsilon}$, that is, $\gamma_i\leq\frac{1}{1-\frac{1}{(1+\epsilon)(1-\xi)}}$, we have
\begin{eqnarray}
i< \frac{2}{\epsilon-\xi-\epsilon\xi}=\frac{2}{(1-\frac{1+\epsilon}{\epsilon}\xi)\epsilon}.\label{for-bc2-5}
\end{eqnarray}

Consequently, we obtain the theorem.

\section{$k$-Center Clustering with Outliers}
\label{sec-ext-kcenter}

%
%


Let $\gamma\in(0,1)$. Given a set $P$ of $n$ points in $\mathbb{R}^d$, the problem of  $k$-center clustering with outliers is to find $k$ balls to cover $(1-\gamma)n$ points, and the maximum radius of the balls is minimized ({\em w.l.o.g.,} we can assume that the $k$ balls have the same radius). 
Given an instance $(P, \gamma)$, let $\{C_1, \cdots, C_k\}$ be the $k$ clusters forming $P_{opt}$ (the subset of $P$ yielding the optimal solution), and $r_{opt}$ be the optimal radius; that is, each $C_j$ is covered by an individual ball with radius $r_{opt}$. Similar to Section~\ref{sec-outlier-general}, we first introduce a linear time algorithm, and then show how to modify it to be sub-linear time by using Lemma~\ref{lem-outlier-general1-general} and \ref{lem-outlier-general2-generalize}.

%


\textbf{Linear time algorithm.} Our algorithm in Section~\ref{sec-quality} can be generalized to be a linear time bi-criteria algorithm for the problem of $k$-center clustering with outliers, if $k$ is assumed to be a constant. Our idea is as follows. In Algorithm~\ref{alg-outlier}, we maintain a set $T$ as the core-set of $P_{opt}$; here, we instead maintain $k$ sets $T_1, T_2, \cdots, T_k$ as the core-sets of $C_1, C_2, \cdots, C_k$, respectively. Consequently, each $T_j$ for $1\leq j\leq k$ has an approximate MEB center $o^j_i$ in the $i$-th round of Step 3, and we let $O_i=\{o^1_i, \cdots, o^k_i\}$. Initially, $O_0$ and $T_j$ for $1\leq j\leq k$ are all empty; we randomly select a point $p\in P$, and with probability $1-\gamma$, $p\in P_{opt}$ (w.l.o.g., we assume $p\in C_1$ and add it to $T_1$; thus $O_1=\{p\}$ after this step). 
We let $Q$ be the set of farthest $t=(1+\delta)\gamma n$ points to $O_i$, and $l_i$ be the $(t+1)$-th largest distance from $P$ to $O_i$ (the distance from a point $p\in P$ to $O_i$ is $\min_{1\leq j\leq k}||p-o^j_i||$). Then, we randomly select a point $q\in Q$, and with probability $\frac{\delta}{1+\delta}$, $q\in P_{opt}$ (as (\ref{for-lem-outlier-1}) in Lemma~\ref{lem-outlier-1}). For ease of presentation, we assume that $q\in P_{opt}$ happens and we have an ``oracle'' to guess which optimal cluster $q$ belongs to, say $q\in C_{j_q}$; then, we add $q$ to $T_{j_q}$ and update the approximate MEB center of $T_{j_q}$. Since each optimal cluster $C_j$ for $1\leq j\leq k$ has the core-set with size $\frac{2}{\epsilon}+1$ (by setting $s=\frac{\epsilon}{2+\epsilon}$ in Theorem~\ref{the-newbc}), after adding at most $k(\frac{2}{\epsilon}+1)$ points, the distance $l_i$ will be smaller than $(1+\epsilon)r_{opt}$. Consequently, a $(1+\epsilon, 1+\delta)$-approximation solution is obtained when $i\geq k(\frac{2}{\epsilon}+1)$. \textbf{Note} that some ``small'' clusters could be missing from the above random sampling based approach and therefore $|O_i|$ could be less than $k$; however, it always can be guaranteed that the total number of missing inliers is at most $\delta\gamma n$, {\em i.e.,} a $(1+\epsilon, 1+\delta)$-approximation is always guaranteed (otherwise, the ratio $\frac{|P_{opt}\cap Q|}{|Q|}>\frac{\delta}{1+\delta}$ and we can continue to sample a point from $P_{opt}$ so as to update $O_i$). 

To remove the oracle for guessing the cluster containing $q$, we can enumerate all the possible $k$ cases; since we add $k(\frac{2}{\epsilon}+1)$ points to $T_1, T_2, \cdots, T_k$, it generates $k^{k(\frac{2}{\epsilon}+1)}=2^{k\log k (\frac{2}{\epsilon}+1)}$ solutions in total, and at least one yields  a $(1+\epsilon, 1+\delta)$-approximation with probability $(1-\gamma)(\frac{\delta}{1+\delta})^{k(\frac{2}{\epsilon}+1)}$ (by the same manner for proving Theorem~\ref{the-outlier}).

\begin{theorem}
\label{the-kcenter}
Let $(P, \gamma)$ be an instance of $k$-center clustering with outliers. Given two parameters $\epsilon, \delta\in (0,1)$, there exists an algorithm that outputs a $(1+\epsilon, 1+\delta)$-approximation with probability $(1-\gamma)(\frac{\delta}{1+\delta})^{k(\frac{2}{\epsilon}+1)}$. The running time is $O(2^{k\log k (\frac{2}{\epsilon}+1)}(n+\frac{1}{\epsilon^5})d)$.

If one repeatedly runs the algorithm $O(\frac{1}{1-\gamma}(1+\frac{1}{\delta})^{k(\frac{2}{\epsilon}+1)})$ times, with constant probability, the algorithm outputs a $(1+\epsilon,1+\delta)$-approximation solution.

\end{theorem}

Similar to our discussion on the running time for MEB with outliers in Section~\ref{sec-quality}, B\u{a}doiu {\em et al.}~\cite{BHI} also achieved a linear time bi-criteria approximation for the $k$-center clustering with outliers problem (see Section 4 in their paper). However, the hidden constant of their running time is exponential in $(\frac{k}{\epsilon\mu})^{O(1)}$ (where $\mu$ is defined in ~\cite{BHI}, and should be $\delta\gamma$ to ensure a $(1+\epsilon, 1+\delta)$-approximation) that is much larger than ``$k\log k (\frac{2}{\epsilon}+1)$'' in Theorem~\ref{the-kcenter}. 

\vspace{0.05in}
\textbf{Sub-linear time algorithm.} The linear time algorithm can be further improved to be sub-linear time; the idea is similar to that for designing sub-linear time algorithm for MEB with outliers in Section~\ref{sec-oulier-improve}. First, we follow Definition~\ref{def-outlier-general} and define the shape set $\mathcal{X}$, where each $x\in\mathcal{X}$ is union of $k$ balls in the space; 
 the center $c(x)$ should be the set of its $k$ ball centers, say $c(x)=\{o^1_x, o^2_x, \cdots, o^k_x\}$, and the size $s(x)$ is the radius, {\em i.e.,} $x=\cup^k_{j=1}\mathbb{B}(o^j_x, s(x))$. Obviously, if $x$ is a feasible solution for the instance $(P, \gamma)$, $\Big|P\cap (\cup^k_{j=1}\mathbb{B}(o^j_x, s(x)))\Big|$ should be at least $(1-\gamma)n$. Also, define the distance function $f(c(x), p)=\min_{1\leq j\leq k}||p-o^j_x||$. It is easy to verify that the shape set $\mathcal{X}$ satisfies Property~\ref{prop-1}, \ref{prop-2}, and \ref{prop-4} defined in Section~\ref{sec-ext}.
From Lemma~\ref{lem-outlier-general1-general}, we know that it is possible to obtain a point in $P_{opt}\cap Q$ with probability $(1-\eta_1) \frac{\delta}{3(1+\delta)}$. Further, we can estimate the value $l_i$ and select the best candidate solution based on Lemma~\ref{lem-outlier-general2-generalize}. Overall, we have the following theorem.


\begin{theorem}
\label{the-kcenter2}
Let $(P, \gamma)$ be an instance of $k$-center clustering with outliers. Given the parameters $\epsilon, \delta, \eta_1, \eta_2\in (0,1)$, there exists an algorithm that outputs a $(1+\epsilon, 1+O(\delta))$-approximation with probability $(1-\gamma)\big((1-\eta_1)(1-\eta_2)\frac{\delta}{3(1+\delta)}\big)^{k(\frac{2}{\epsilon}+1)}$. The running time is $\tilde{O}(2^{k\log k (\frac{2}{\epsilon}+1)}(\frac{1}{\delta^2\gamma}+\frac{1}{\epsilon^5})d)$.

If one repeatedly runs the algorithm  $N=O\Big(\frac{1}{1-\gamma}\big(\frac{1}{1-\eta_1}(3+\frac{3}{\delta})\big)^{k(\frac{2}{\epsilon}+1)}\Big)$ times with setting $\eta_2=O(\frac{1}{2^{k\log k (\frac{2}{\epsilon}+1)}N})$, with constant probability, the algorithm outputs a $(1+\epsilon,1+O(\delta))$-approximation solution.

\end{theorem}

\section{Flat Fitting with Outliers}
\label{sec-flat}
Let $j$ be a fixed integer between $0$ and $d$. Given a $j$-dimensional flat $\mathcal{F}$ and a point $p\in\mathbb{R}^d$, we define their distance, $dist(\mathcal{F}, p)$, to be the Euclidean distance from $p$ to its projection onto $\mathcal{F}$. Let $P$ be a set of $n$ points in $\mathbb{R}^d$. The problem of flat fitting is to find the $j$-dimensional flat $\mathcal{F}$ that minimizes $\max_{p\in P}dist(\mathcal{F}, p)$. It is easy to see that the MEB problem is the case $j=0$ of the flat fitting problem. Furthermore, given a parameter $\gamma\in (0,1)$, the flat fitting with outliers problem is to find a subset $P'\subset P$ with size $(1-\gamma)n$ such that  $\max_{p\in P'}dist(\mathcal{F}, p)$ is minimized. Similar to MEB with outliers, we also use $P_{opt}$ to denote the optimal subset. Before presenting our algorithms for flat fitting with outliers, we first introduce the linear time algorithm from Har-Peled and Varadarajan~\cite{DBLP:journals/dcg/Har-PeledV04} for the vanilla version (without outliers). 

We start from the case $j=1$, {\em i.e.,} the flat $\mathcal{F}$ is a line in the space. Roughly speaking, their algorithm is an iterative procedure to update the solution round by round, until it is close enough to the optimal line $l_{opt}$. There are two parts in the algorithm. \textbf{(1)} It picks an arbitrary point $p_\Delta\in P$ and let $q_\Delta$ be the farthest point of $P$ from $p_\Delta$; it can be proved that the line passing through $p_\Delta$ and $q_\Delta$, denoted as $l_0$, is a good initial solution that yields a $4$-approximation with respect to the objective function. \textbf{(2)} In each of the following  rounds, the algorithm updates the solution from $l_{i-1}$ to $l_i$ where $i\geq 1$ is the current number of rounds: let $p_i$ be the farthest point of $P$ from $l_{i-1}$ and let $h_i$ denote the $2$-dimensional flat spanned by $p_i$ and $l_{i-1}$; then the algorithm computes a set of $O(\frac{1}{\epsilon^8}\log^2\frac{1}{\epsilon})$ lines on $h_i$, and picks one of them as $l_i$ via an ``oracle''. They proved that the improvement from $l_{i-1}$ to $l_i$ is significant enough; thus, after running $\nu=O(\frac{1}{\epsilon^3}\log\frac{1}{\epsilon})$ rounds, it is able to achieve a $(1+\epsilon)$-approximation. To remove the ``oracle'', the algorithm can enumerate all the $O(\frac{1}{\epsilon^8}\log^2\frac{1}{\epsilon})$ lines on $h_i$, and thus the total running time is $O\big(2^{\frac{1}{\epsilon^3}\log^2\frac{1}{\epsilon}} n d\big)$. 

\vspace{0.05in}

\textbf{Linear time algorithm.} Now we consider to adapt the above algorithm to the case with outliers, where in fact the idea is similar to the idea proposed in Section~\ref{sec-quality} for MEB with outliers. For simplicity, we still use the same notations as above. Consider the part \textbf{(1)} first. If we randomly pick a point $p_{\Delta}$ from $ P$, with probability $1-\gamma$, it belongs to $P_{opt}$; further, we randomly pick a point, denoted as $q_\Delta$, from the set of $(1+\delta_0)\gamma n$ farthest points of $P$ from $p_{\Delta}$, where the value of $\delta_0$ will be determined below. Obviously, with probability $\frac{\delta_0}{1+\delta_0}$, $q_\Delta\in P_{opt}$. Denote by $P_0=\{p\in P_{opt}\mid ||p-p_\Delta||\leq ||q_\Delta-p_\Delta||\}$. Then we have the following lemma.

\begin{lemma}
\label{lem-flat-1}
Denote by $l_0$ the line passing through $p_\Delta$ and $q_\Delta$. Then, with probability $(1-\gamma)(\frac{\delta_0}{1+\delta_0})$, 
\begin{eqnarray}
\max_{p\in P_0}dist(l_0, p)\leq 4 \max_{p\in P_{0}}dist(l_{opt}, p)\leq 4 \max_{p\in P_{opt}}dist(l_{opt}, p).\label{for-flat-1}
\end{eqnarray}
Also, the size of $P_0$ is at least $\big(1-(1+\delta_0)\gamma\big)n$. 
\end{lemma}
It is straightforward to obtain the size of $P_0$. The inequality (\ref{for-flat-1}) directly comes from the aforementioned result of~\cite{DBLP:journals/dcg/Har-PeledV04}, as long as $p_\Delta$ and $q_\Delta\in P_{opt}$. So we can use the line $l_0$ as our initial solution. Then, we can apply the same random sampling idea to select the point $p_i$ in the $i$-th round. Namely, we randomly pick a point as $p_i$ from the set of $(1+\delta_0)\gamma n$ farthest points of $P$ from $l_i$. Moreover, we need to shrink the set $P_{i-1}$ to $P_i=\{p\in P_{i-1}\mid dist(l_{i-1}, p)\leq dist(l_{i-1}, p_i)\}$. Similar to Lemma~\ref{lem-flat-1}, we can show that the improvement from $l_{i-1}$ to $l_{i}$ is significant enough with probability $(1-\gamma)(\frac{\delta_0}{1+\delta_0})^{i+1}$, and the size of $P_i$ is at least $\big(1-(1+(i+1)\delta_0)\gamma\big) n$. After running $\nu$ rounds, we obtain the line $l_\nu$ such that $\max_{p\in P_\nu}dist(l_\nu, p)\leq (1+\epsilon)  \max_{p\in P_{opt}}dist(l_{opt}, p)$, and $|P_\nu|\geq \big(1-(1+(\nu+1)\delta_0)\gamma\big)n$. So if we set $\delta_0=\frac{\delta}{\nu+1}$ with a given $\delta\in (0,1)$, the line $l_\nu$ will be a bi-criteria $(1+\epsilon, 1+\delta)$-approximation of the instance $(P, \gamma)$. By using the idea in~\cite{DBLP:journals/dcg/Har-PeledV04}, we can extend the result to the case $j>1$ with $\nu=\frac{e^{O(j^2)}}{\epsilon^{2j+1}}\log\frac{1}{\epsilon}$. We refer the reader to~\cite{DBLP:journals/dcg/Har-PeledV04} for more details. 

\begin{theorem}
\label{the-flat}
Let $(P, \gamma)$ be an instance of $j$-dimensional flat fitting with outliers. Given two parameters $\epsilon, \delta\in (0,1)$, there exists an algorithm that outputs a $(1+\epsilon, 1+\delta)$-approximation with probability $(1-\gamma)\big(\frac{1}{2}\big)^{g(j, \epsilon)}$ where $g(j, \epsilon)=poly(e^{O(j^2)}, \frac{1}{\epsilon^j})$. The running time is $O(2^{g'(j, \epsilon)}nd)$ where $g'(j, \epsilon)=poly(e^{O(j^2)}, \frac{1}{\epsilon^j})$.

If one repeatedly runs the algorithm $ \frac{2^{g(j, \epsilon)}}{1-\gamma}$ times, with constant probability, the algorithm outputs a $(1+\epsilon,1+\delta)$-approximation solution.

\end{theorem}

\textbf{Sub-linear time algorithm.} We can view the flat fitting with outliers problem as an MEX with outliers problem. Let $r\geq 0$ and $\mathcal{F}$ be a $j$-dimensional flat. Then we can define a $j$-dimensional ``slab'' $SL(\mathcal{F}, r)=\{p\in\mathbb{R}^d\mid dist(\mathcal{F}, p)\leq r\}$, where its ``center'' and ``size'' are $\mathcal{F}$ and $r$ respectively ({\em e.g.,} a ball is a $0$-dimensional slab); the distance function $f(\mathcal{F}, p)=dist(\mathcal{F}, p)$. It is easy to see that the shape set of slabs satisfies  Property~\ref{prop-1}, \ref{prop-2}, and \ref{prop-4} defined in Section~\ref{sec-ext}. Furthermore, finding the optimal flat is equivalent to finding the smallest slab covering $(1-\gamma)n$ points of $P$. Therefore, by using  Lemma~\ref{lem-outlier-general1-general} and~\ref{lem-outlier-general2-generalize}, we achieve the following theorem.

\begin{theorem}
\label{the-flat2}
Let $(P, \gamma)$ be an instance of $j$-dimensional flat fitting with outliers. Given the parameters $\epsilon, \delta, \eta_1, \eta_2\in (0,1)$, there exists an algorithm that outputs a $(1+\epsilon, 1+O(\delta))$-approximation with probability $(1-\gamma)\big((1-\eta_1)(1-\eta_2)\frac{\delta}{3(1+\delta)}\big)^{g(j, \epsilon)}$ where $g(j, \epsilon)=poly(e^{O(j^2)}, \frac{1}{\epsilon^j})$. The running time is $O(2^{g'(j, \epsilon,\delta,\gamma)} d)$ where $g'(j, \epsilon)=poly(e^{O(j^2)}, \frac{1}{\epsilon^j}, \frac{1}{\delta}, \frac{1}{\gamma})$.

If one repeatedly runs the algorithm  $N=O\Big(\frac{1}{1-\gamma}\big(\frac{1}{1-\eta_1}(3+\frac{3}{\delta})\big)^{g(j, \epsilon)}\Big)$ times with setting $\eta_2=O(\frac{1}{2^{g(j, \epsilon)}N})$, with constant probability, the algorithm outputs a $(1+\epsilon,1+O(\delta))$-approximation solution.

\end{theorem}

\section{Support Vector Machine with Outliers}
\label{sec-svm}

In practice, datasets often contain outliers.  The separating margin of SVM  could be considerably deteriorated by outliers.  As mentioned in~\cite{ding2015random}, most of existing techniques~\cite{conf/aaai/XuCS06,icml2014c2_suzumura14}  for SVM outliers removal are numerical approaches ({\em e.g.,} adding some penalty item to the objective function), and only can guarantee local optimums. Ding and Xu~\cite{ding2015random} modeled SVM with outliers as a combinatorial optimization problem and provided an algorithm called ``Random Gradient Descent Tree''. 
We focus on one-class SVM with outliers first, and explain the extension for two-class SVM with outliers at the end of this section.

\subsection{One-class SVM with Outliers}

Below is the definition of the one-class SVM with outliers problem proposed in~\cite{ding2015random}. 

\begin{definition} [One-class SVM with Outliers]
\label{def-svm}
 Given a set $P$ of $n$ points in $\mathbb{R}^d$ and a small parameter $\gamma\in (0,1)$, the one-class SVM with outliers problem is to find a subset $P'\subset P$ with size $(1-\gamma)n$ and a hyperplane $\mathcal{H}$ separating the origin $o$ and $P'$, such that the distance between $o$ and $\mathcal{H}$ is maximized. 
%
\end{definition}

\textbf{Linear time algorithm.} We briefly overview the algorithm of~\cite{ding2015random}. They also considered the ``bi-criteria approximation'' with two small parameters $\epsilon, \delta\in (0,1)$: a hyperplane  $\mathcal{H}$ separates the origin $o$ and a subset $P'\subset P$ with size $\big(1-(1+\delta)\gamma\big)n$, where the distance between $o$ and $\mathcal{H}$ is at least $(1-\epsilon)$ of the optimum. The idea of~\cite{ding2015random} is based on the fact that the SVM (without outliers) problem is equivalent to the polytope distance problem in computational geometry~\cite{GJ09}. 

{\em Let $o$ be the origin and  $P$ be a given set of points in $\mathbb{R}^d$. The \textbf{polytope distance problem} is to find a point $q$ inside the convex hull of $P$ so that the distance $||q-o||$ is minimized. }

For an instance $P$ of one-class SVM, it can be proved that the vector $q_{opt}-o$, if $q_{opt}$ is the optimal solution for the polytope distance between $o$ and $P$, is the normal vector of the optimal hyperplane. We refer the reader to \cite{ding2015random,GJ09} for more details. The polytope distance problem can be efficiently solved by {\em Gilbert Algorithm}~\cite{frank1956algorithm,Gilbert:1966:IPC}. For completeness, we present it in Algorithm~\ref{alg-gilbert}.

Similar to the core-set construction method of MEB in Section~\ref{sec-newanalysis}, the algorithm also greedily improves the current solution by selecting some point $p_i$ in each iteration. Let $\rho$ be the polytope distance between $o$ and $P$, $D=\max_{p,q\in P}||p-q||$, and $E=\frac{D^2}{\rho^2}$. Given $\epsilon\in (0,1)$, it has been proved that a $(1-\epsilon)$-approximation of one-class SVM ({\em i.e.,} a separating margin with the width at least $(1-\epsilon)$ of the optimum) can be achieved by running Algorithm~\ref{alg-gilbert}  at most $2\lceil 2E/\epsilon\rceil$ steps~\cite{GJ09,C10}. To handle outliers, the algorithm of~\cite{ding2015random} follows the similar intuition of Section~\ref{sec-quality}; it replaces the step of greedily selecting the point $p_i$ by randomly sampling a point from a set $Q$, which contains the $(1+\delta)\gamma n$ points  having the smallest projection distances ({\em i.e.,} the values of the function $\frac{\langle p, v_i\rangle}{||v_i||}$ in Step 2(a) of Algorithm~\ref{alg-gilbert}). To achieve a $(1-\epsilon, 1+\delta)$-approximation with constant success probability,  the algorithm takes $O\big(\frac{1}{1-\gamma}(1+\frac{1}{\delta})^{z}\frac{D^2}{\epsilon\rho^2}nd\big)$ time, where $z=O(\frac{D^2}{\epsilon\rho^2})$. 

\begin{algorithm}[tb]
   \caption{Gilbert Algorithm \cite{Gilbert:1966:IPC,ding2015random}}
   \label{alg-gilbert}
\begin{algorithmic}
   \STATE {\bfseries Input:} A point-set $P$ in $\mathbb{R}^d$, and $N\in \mathbb{Z}^+$.
    \STATE {\bfseries Output:} $v_i$ as an approximate solution of the polytope distance between the origin and $P$.
    \begin{enumerate}
   \item Initialize $i=1$ and $v_1$ to be the closest point in $P$ to the origin $o$.
   \item Iteratively perform the following steps until $i=N$.
   \begin{enumerate}
   \item Find the point $p_i \in P$ whose orthogonal projection on the supporting line of segment $\overline{ov_{i}}$ has the closest distance to $o$ (called the projection distance of $p_{i}$), {\em i.e.,} $p_i=\arg\min_{p\in P}\{\frac{\langle p, v_i\rangle}{||v_i||}\}$, 
     where $\langle p, v_i\rangle$ is the inner product of $p$ and $v_i$ (see Figure~\ref{fig-gilbert}).
   \item Let $v_{i+1}$ be the point on segment $\overline{v_i p_i}$ closest to the origin $o$; update $i=i+1$.
    \end{enumerate}
    \end{enumerate}
\end{algorithmic}
\end{algorithm}

\begin{figure}[]
\vspace{-0.1in}
   \centering
  \includegraphics[height=1.2in]{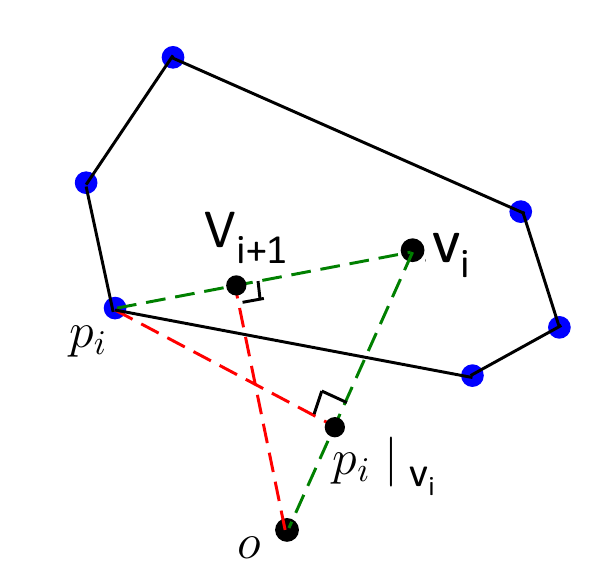}
  \vspace{-0.1in}
      \caption{An illustration of step 2 in Algorithm~\ref{alg-gilbert}; $p_i\mid_{v_i}$ is the projection of $p_i$ on $\overline{o v_i }$.}
  \label{fig-gilbert}
  \vspace{-0.1in}
\end{figure}



\begin{figure}[]
\vspace{-0.1in}
   \centering
  \includegraphics[height=1.2in]{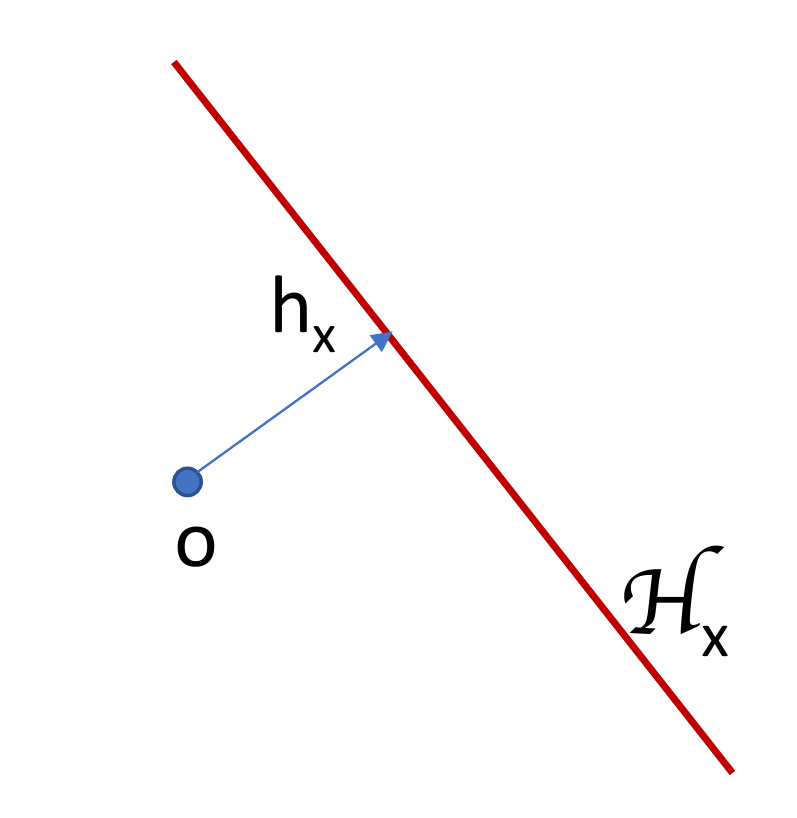}
  \vspace{-0.1in}
      \caption{An illustration for $\mathcal{H}_x$ and $h_x$.}
  \label{fig-nsvm1}
  \vspace{-0.1in}
\end{figure}

\textbf{Sub-linear time algorithm.}
We define $\mathcal{X}$ to be the set of all the closed half-spaces not covering the origin $o$ in $\mathbb{R}^d$; for each $x\in \mathcal{X}$, let $\mathcal{H}_x$ be the hyperplane enclosing $x$ and let $h_x$ be the projection of $o$ on $\mathcal{H}_x$ (see Figure~\ref{fig-nsvm1}). We suppose that the given instance $(P, \gamma)$ has feasible solution. That is, there exists at least one half-space $x\in \mathcal{X}$ that the hyperplane $\mathcal{H}_x$ separates the origin $o$ and a subset $P'$ with size $(1-\gamma)n$. We define the center $c(x)= \frac{h_x}{||h_x||}$; since the MEX with outlier problem in Definition~\ref{def-outlier-general} is a minimization problem, we design the size function  $s(x)= \frac{1}{||h_x||}$. 
%
Obviously, a $(1-\epsilon)$-approximation of the SVM with outliers problem  is equivalent to a $\frac{1}{1-\epsilon}$-approximation with respect to the size function $s(x)$. 
 We also define the distance function $f(c(x), p)=-\langle p, \frac{h_x}{||h_x||}\rangle$. It is easy to verify that the shape set $\mathcal{X}$ satisfies Property~\ref{prop-1}, \ref{prop-2}, and~\ref{prop-4}. 
 
Recall that Algorithm~\ref{alg-gilbert} selects the point $p_i=\arg\min_{p\in P}\{\frac{\langle p, v_i\rangle}{||v_i||}\}$ in each iteration. Actually, the vector $\frac{v_i}{||v_i||}$ can be viewed as a shape center 
and $p_i$  is the farthest point to $\frac{v_i}{||v_i||}$ based on the distance function $f(c(x), p)$. Moreover, the set $Q$ mentioned in the previous linear time algorithm actually is the set of the farthest $(1+\delta)\gamma n$ points from $P$ to $\frac{v_i}{||v_i||}$. 
Consequently, we can apply  Lemma~\ref{lem-outlier-general1-general} to sample a point from $P_{opt}\cap Q$, and apply Lemma~\ref{lem-outlier-general2-generalize} to estimate the value of $l_i$ for each candidate solution $\frac{v_i}{||v_i||}$. Overall, we can improve the running time of the algorithm of~\cite{ding2015random} to be independent of $n$. 

\begin{theorem}
\label{the-svm}
Let $(P, \gamma)$ be an instance of SVM with outliers. Given the parameters $\epsilon, \delta, \eta_1, \eta_2\in (0,1)$, there exists an algorithm that outputs a $(1-\epsilon, 1+O(\delta))$-approximation with probability $(1-\gamma)\big((1-\eta_1)(1-\eta_2)\frac{\delta}{3(1+\delta)}\big)^{z}$ where $z=O(\frac{D^2}{\epsilon\rho^2})$. The running time is $\tilde{O}(\frac{D^2}{\delta^2\gamma\epsilon^2\rho^2} d)$.

If one repeatedly runs the algorithm  $N=O\Big(\frac{1}{1-\gamma}\big(\frac{1}{1-\eta_1}(3+\frac{3}{\delta})\big)^{z}\Big)$ times with setting $\eta_2=O(\frac{1}{zN})$, with constant probability, the algorithm outputs a $(1-\epsilon,1+O(\delta))$-approximation solution.

\end{theorem}

\subsection{Two-class SVM with Outliers}
 
 Below is the definition of the two-class SVM with outliers problem proposed in~\cite{ding2015random}. 

\begin{definition} [Two-class SVM with Outliers]
\label{def-svm2}
 Given two point sets $P_1$ and $P_2$ in $\mathbb{R}^d$ and two small parameters $\gamma_1, \gamma_2\in (0,1)$, the two-class SVM with outliers problem is to find two subsets $P'_1\subset P_1$ and $P'_2\subset P_2$ with $|P'_1|=(1-\gamma_1)|P_1|$ and $|P'_2|=(1-\gamma_2)|P_2|$, and a margin separating  $P'_1$ and $P'_2$, such that the width of the margin is maximized.   
 %
\end{definition}
We use $P^{opt}_1$ and $P^{opt}_2$, where $|P^{opt}_1|=(1-\gamma_1)|P_1|$ and $|P^{opt}_2|=(1-\gamma_2)|P_2|$, to denote the subsets of $P_1$ and $P_2$ which are separated by the optimal margin. 
The ordinary two-class SVM (without outliers) problem is equivalent to computing the polytope distance between the origin $o$ and $\mathcal{M}(P_1, P_2)$, where  $\mathcal{M}(P_1, P_2)$ is the Minkowski difference of $P_1$ and $P_2$~\cite{GJ09}. Note that it is not necessary to compute the set  $\mathcal{M}(P_1, P_2)$ explicitly. Instead, Algorithm~\ref{alg-gilbert} only needs to select one point from $\mathcal{M}(P_1, P_2)$ in each iteration, and overall the running time is still linear in the input size. To deal with two-class SVM with outliers, Ding and Xu~\cite{ding2015random} slightly modified their algorithm for the case of one-class. In each iteration, it considers two subsets $Q_1\subset P_1$ and $Q_2\subset P_2$, which respectively consist of points having the $(1+\delta)\gamma_1|P_1|$ smallest  projection distances among all points in $P_{1}$ and the $(1+\delta)\gamma_2|P_2|$ largest  projection distances among all points in $P_2$ on the vector $v_i$; then, the algorithm randomly selects two points $p^1_i\in Q_1$ and $p^2_i\in Q_2$, and their difference vector $p^2_i-p^1_i$ will serve as the role of $p_i$ in Step 2(a) of  Algorithm~\ref{alg-gilbert} to update the current solution $v_i$. This approach yields a $(1-\epsilon, 1+\delta)$-approximation in linear time.

\begin{figure}[]
\vspace{-0.1in}
   \centering
  \includegraphics[height=2in]{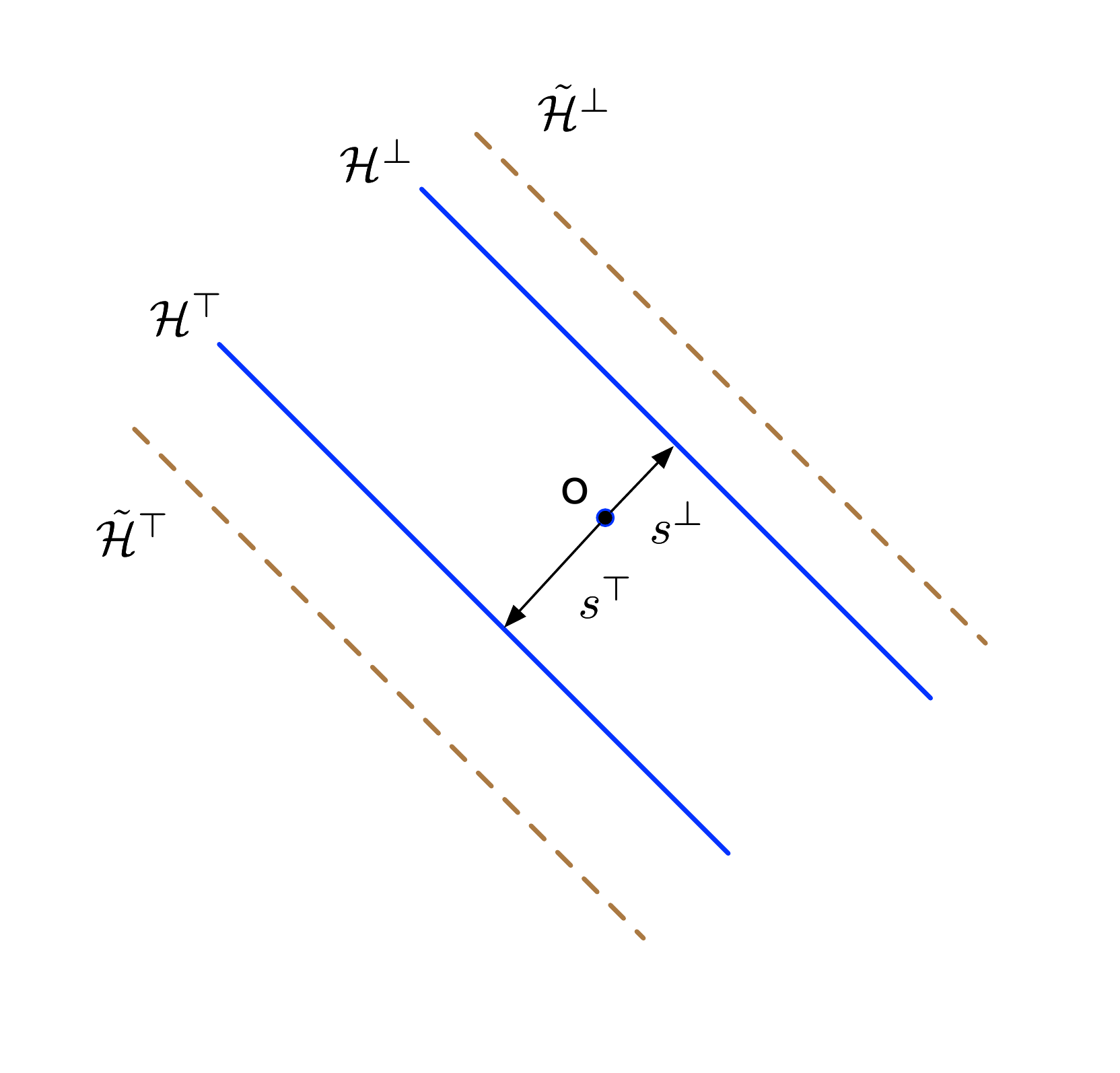}
  \vspace{-0.1in}
      \caption{An illustration for two-class SVM. The distances from $o$ to $\mathcal{H}^\perp$ and $\mathcal{H}^\top$ are $s^\perp$ and $s^\top$, respectively. The hyperplanes  $\tilde{\mathcal{H}}^\perp$ and $\tilde{\mathcal{H}}^\top$ are the estimations of $\mathcal{H}^\perp$ and $\mathcal{H}^\top$, and the distances from $o$ to them are $\tilde{s}^\perp$ and $\tilde{s}^\top$ respectively.}
  \label{fig-twosvm}
  \vspace{-0.1in}
\end{figure}

To improve the algorithm to be sub-linear, we need several modifications on our previous idea for the case of one-class. First, we change the distance function to be:
\[ f(p, c)= \left\{ \begin{array}{ll}
       -\langle p, \frac{h_x}{||h_x||}\rangle & \mbox{if $p\in P_1$};\\
       \langle p, \frac{h_x}{||h_x||}\rangle& \mbox{if $p\in P_2$}.\end{array} \right. \] 
By using this new distance function, we can apply Lemma~\ref{lem-outlier-general1-general} to obtain the points $p^1_i\in Q_1\cap P^{opt}_1$ and $p^2_i\in Q_2\cap P^{opt}_2$ separately in sub-linear time. Given a vector ({\em i.e.,} candidate center) $\frac{v_i}{||v_i||}$, assume $\mathcal{H}^\perp$ and $\mathcal{H}^\top$ are the parallel hyperplanes orthogonal to $\frac{v_i}{||v_i||}$ that the margin formed by them separates $P'_1$ and $P'_2$, where $P'_1\subset P_1$ and $P'_2\subset P_2$ with $|P'_1|=(1-\gamma_1)|P_1|$ and $|P'_2|=(1-\gamma_2)|P_2|$. 
Without loss of generality, we assume that the origin $o$ is inside the margin. Suppose that the distances from $o$ to $\mathcal{H}^\perp$ and $\mathcal{H}^\top$ are $s^\perp$ and $s^\top$, respectively. Then, we obtain two shapes (closed half-spaces) $x^\perp=(-\frac{v_i}{||v_i||}, \frac{1}{s^\perp})$ and $x^\top=(\frac{v_i}{||v_i||}, \frac{1}{s^\top})$ with $P'_1\subset x^\perp$ and $P'_2\subset x^\top$. Consequently, we can apply  Lemma~\ref{lem-outlier-general2-generalize} twice to obtain two values $ \frac{1}{\tilde{s}^\perp}\leq \frac{1}{s^\perp}$ and $ \frac{1}{\tilde{s}^\top}\leq \frac{1}{s^\top}$ with $\Big| P_1\setminus x(-\frac{v_i}{||v_i||}, \frac{1}{\tilde{s}^\perp})\Big|\leq (1+O(\delta))\gamma_1|P_1|$ and $\Big| P_2\setminus x(\frac{v_i}{||v_i||}, \frac{1}{\tilde{s}^\top})\Big|\leq (1+O(\delta))\gamma_2|P_2|$. Therefore, we can use the value $\tilde{s}^\perp+\tilde{s}^\top$ as an estimation of $s^\perp+s^\top$. See Figure~\ref{fig-twosvm} for an illustration. Overall, we can achieve a $(1-\epsilon, 1+O(\delta))$-approximation in sub-linear time.

\end{document}